\documentclass[runningheads]{llncs}
\usepackage{graphicx}
\usepackage{xcolor}

\newif\ifdraft \drafttrue
\newif\iffull \fulltrue
\newif\ifcomment \commentfalse
\draftfalse
\usepackage{wibrown}

\begin{document}
\title{Targeted Intervention in Random Graphs
}
%
%
\author{William Brown
\and
Utkarsh Patange 
}
\authorrunning{W. Brown and U. Patange}
%
\institute{Columbia University, New York NY 10027, USA\\
\email{\{w.brown, utkarsh.patange\}@columbia.edu}}

\maketitle

\begin{abstract}
We consider a setting where individuals interact in a network, each choosing actions which optimize utility as a function of neighbors' actions. A central authority aiming to maximize social welfare at equilibrium can intervene by paying some cost to shift individual incentives, and the optimal intervention can be computed using the spectral decomposition of the graph, yet this is infeasible in practice if the adjacency matrix is unknown. In this paper, we study the question of designing intervention strategies for graphs where the adjacency matrix is unknown and is drawn from some distribution. For several commonly studied random graph models, we show that there is a single intervention, proportional to the first eigenvector of the expected adjacency matrix, which is near-optimal for almost all generated graphs when the budget is sufficiently large.  We also provide several efficient sampling-based approaches for approximately recovering the first eigenvector when we do not know the distribution.  On the whole, our analysis compares three categories of interventions: those which use no data about the network, those which use some data (such as distributional knowledge or queries to the graph), and those which are fully optimal.  We evaluate these intervention strategies on synthetic and real-world network data, and our results suggest that analysis of random graph models can be useful for determining when certain heuristics may perform well in practice.


\end{abstract}

\section{Introduction}
\label{intro}

Individual decision-making in many domains is driven by personal as
well as social factors. If one wants to decide a level of time, money,
or effort to exert on some task, the behaviors of one's friends or
neighbors can be powerful influencing factors. We can view these
settings as games where agents in a network are playing some game,
each trying to maximize their individual utility as a function of
their ``standalone value'' for action as well as their neighbors'
actions. 
The actions of agents who are ``central'' in a network can have large ripple effects. 
Identifying and understanding the role of central agents is of high importance for tasks ranging from
microfinance
\cite{Banerjee1236498} and vaccinations \cite{10.1093/restud/rdz008},
to tracking the spread of information throughout a community \cite{banerjee14SSRN}. We view our work as providing theoretical support for heuristic approaches to intervention in these settings.

A model for such a setting is studied in recent work by Galeotti, Golub, and Goyal
\cite{Galeotti2017TargetingII}, where they ask the natural question of
how a third party should ``intervene'' in the network to maximally
improve social welfare at the equilibrium of the game. Interventions are modeled by
assuming that the third party can pay some cost to adjust the
standalone value parameter of any agent, and must decide how to
allocate a fixed budget. This may be interpreted as suggesting that
these targeted agents are subjected to advertizing, monetary
incentives, or some other form of encouragement. For their model, they
provide a general framework for computing the optimal intervention
subject to any budget constraint, which can be expressed in terms of
the spectral decomposition of the graph. For large budgets, the
optimal intervention is approximately proportional to the first
eigenvector of the adjacency matrix of the graph, a common measure of
network centrality.

While this method is optimal, and computable in polynomial time if the
adjacency matrix is known, it is rare in practice that we can hope to
map all connections in a large network. For physical networks, edges representing personal connections may be far harder to map than simply identifying the set of agents, and for large digital networks we may be bound by computational or data access constraints. However, real-world networks
are often well-behaved in that their structure can be approximately
described by a simple generating process. If we cannot afford to map an entire network, is optimal targeted intervention feasible at all? 
A natural target would be to implement interventions which are competitive with the optimal intervention, i.e.\ obtaining almost the same increase in social welfare, without access to the full adjacency matrix. 
Under what conditions can we use knowledge of the distribution a graph is drawn from to compute a near-optimal
intervention without observing the realization of the graph? 
Without knowledge of the distribution, how much information about the graph is necessary to find such an intervention?
Can we ever reach near-optimality with no information about the graph?
These are the questions we address.

\subsection{Contributions}

Our main result shows that for random graphs with independent edges,  
the first eigenvector of the ``expected adjacency matrix'', representing the probability of each edge being included in the graph, constitutes a near-optimal intervention simultaneously for almost all generated graphs,
when the budget is large enough and the expected matrix satisfies basic spectral conditions. 
We further explore graphs with given expected degrees, Erd\H{o}s-R\'enyi graphs, power law graphs, and stochastic block model graphs as special cases for which our main result holds. 
In these cases, the first eigenvector of the expected matrix can often be characterized by a  
simple expression of parameters of the graph distribution.


Yet in general, this approach still assumes a fair amount of knowledge about the distribution, and that we can map agents in the network to their corresponding roles in the distribution.
We give several sampling-based methods 
for approximating the first eigenvector
of a graph in each of the aforementioned special cases, which do not assume knowledge of agent identities or distribution parameters, other than group membership in the stochastic block model. 
These methods assume different query models for accessing information about the realized graph, such as the ability to query the existence of an edge or to observe a random neighbor of an agent.
Using the fact that the graph was drawn from {\it some} distribution, we can reconstruct an approximation of the first eigenvector more efficiently than we could reconstruct the full matrix. The lower-information settings we consider can be viewed as assumptions about qualitative domain-specific knowledge, such as a degree pattern which approximately follows an (unknown) power law distribution, or the existence of many tight-knit disjoint communities.

We evaluate our results experimentally on both synthetic and real-world networks for a range of parameter regimes. We find that our heuristic interventions can perform quite well compared to the optimal intervention, even at modest budget and network sizes. These results further illustrate the comparative efficacies of interventions requiring varying degrees of graph information under different values for distribution parameters, budget sizes, and degrees of network effects.

On the whole, our results suggest that explicit mapping of the connections in a network is unnecessary to implement near-optimal targeted interventions in strategic settings, and that distributional knowledge or limited queries will often suffice.

\subsection{Related Work}
Recent work by Akbarpour, Malladi, and Saberi
\cite{Akbarpour:2018:DSV:3219166.3219225} has focused on the challenge
of overcoming network data barriers in targeted interventions under a diffusion model of social influence.
In this setting,
for $\gnp$ and power law random graphs, they derive bounds on the
additional number of ``seeds'' needed to match optimal targeting when
network information is limited. A version of this problem where network information can be purchased is studied in \cite{eckles2019seeding}.
Another similar model was employed by Candogan, Bimpikis, and Ozdaglar \cite{Candogan2012} where they
study optimal pricing strategies to maximize profit of a monopolist
selling service to consumers in a social network where the consumer
experiences a positive local network effect, where notions of centrality play a key role. 
Similar targeting strategies are considered in \cite{Demange2017}, where the planner tries to maximize
aggregate action in a network with complementarities.
\cite{Huang2019TheVO} studies the efficacy of blind interventions in a pricing game for the special case of Erd\H{o}s-R\'enyi graphs.
In \cite{parise2018graphon}, targeted interventions are also studied for ``linear-quadratic games'', quite similar to those from \cite{Galeotti2017TargetingII}, in the setting of infinite-population graphons, where a concentration result is given for near-optimal interventions. 

Our results can be viewed as qualitatively similar findings to the above results
in the model of
\cite{Galeotti2017TargetingII}.  
While they have showed
that exact optimal interventions can be constructed on a graph with
full information, we propose that local information is enough to
construct an approximately optimal intervention for many 
distributions of random graphs. It is argued in
\cite{Breza2017} that collecting data of this kind 
(aggregate relational data)
is easier in real networks compared to obtaining full network
information.
We make use of concentration inequalities for the spectra of random adjacency matrices; there is a great deal of work studying various spectral properties of random graphs (see e.g.\ \cite{Chung2003,chungradcliffe,Chung2002,10.1145/335305.335326,RePEc:arx:papers:1709.10402}). Particularly relevant to us is \cite{RePEc:arx:papers:1709.10402}, which characterizes the asymptotic distributions of various centrality measures for random graphs. There is further relevant literature for studying centrality in graphons, see e.g. \cite{DBLP:journals/corr/Avella-MedinaPS17}.
Of relevance to our sampling techniques, a method for estimating eigenvector centrality via sampling is given in \cite{ruggeri2019sampling}, and the task of finding a ``representative'' sample of a graph is discussed in \cite{10.1145/1150402.1150479}. 

\section{Model and Preliminary Results}\label{prelims}

Here we introduce the ``linear-quadratic'' network game setting from \cite{Galeotti2017TargetingII}, also studied in e.g. \cite{parise2018graphon}, which captures the dynamics of personal and social motivations for action in which we are interested.
\subsection{Setting}
Agents are playing a game on an undirected graph with adjacency matrix $A$.
Each agent takes an action $a_i \in \mathbb{R}$ and obtains individual utility given by:
\begin{align*}
    u_i(a, A) =&\; b_i a_i - \frac{1}{2}a_i^2 + \beta \sum_{j} A_{ij}a_{i}a_{j}
\end{align*}    
Here, $b_i$ represents agent $i$'s ``standalone marginal value'' for action. The parameter $\beta$ controls the effect of strategic dynamics, where a positive sign promotes complementary behavior with neighbors and a negative sign promotes acting in opposition to one's neighbors. In this paper we focus on the case where each value $A_{ij}$ is in $\{0,1\}$ and $\beta > 0$. 
The assumption that $\beta>0$ corresponds to the case where agents' actions are complementary, meaning that an
increase in action by an agent will increase their neighbors'
propensities for action.
\footnote{When $\beta < 0$, neighbors' actions act as substitutes, and one obtains less utility when neighbors increase levels of action. In that case, the optimal intervention for large budgets is approximated by the last eigenvector of the graph, which measures its ``bipartiteness''.}
We assume that $b_i \geq 0$ for each agent as well. 

The matrix $M = (I - \beta A)$ can be used to determine the best response for each agent given their opponents' actions. The best response vector $a^*$, given current actions $a$, can be computed as:  
\begin{align*}
a^* = b + \beta A a.
\end{align*}
Upon solving for $a^* = a$, we get that
$a^* = (I - \beta A)^{-1} b = M^{-1} b$, giving us the Nash
equilibrium for the game as all agents are simultaneously best
responding to each other. We show in \Cref{best-response} that when agents begin with null action values, repeated best responses will converge to equilibrium, and further that the new equilibrium is likewise reached after intervention.

Our results will apply to cases where all eigenvalues of $M$ are
almost surely positive, ensuring invertibility.
\footnote{If $\beta > 0$ and $M$ is not invertible, equilibrium actions will be infinite for all agents in some component of the graph.}
The social welfare of the game $W = \sum_i u_i$ can be computed as a function of the equilibrium actions:
\begin{align*}
    W = \frac{1}{2} (a^*)^{\top} a^*
\end{align*}
Given the above assumptions, equilibrium actions $a^*_i$ will always be non-negative. 
    
\subsection{Targeted Intervention}
In this game, a central authority has the ability to modify agents'
standalone marginal utilities from $b_i$ to $\hat{b}_i$ by paying a
cost of $(b_i - \hat{b}_i)^2$, and their goal is to maximize social
welfare subject to a budget constraint $C$:
$$\max \sum_i u_i \quad  \text{subject to}  \quad \sum_i (b_i - \hat{b}_i)^2 \leq C.$$
Here, an {\it intervention} is a vector $y = \hat{b}- b$ such that
$\norm{y}^2 \leq C$. Let $W(y)$ denote the social welfare at
equilibrium following an intervention $y$.  
\footnote{Unless specified otherwise, $\norm{\cdot}$ refers to the $\ell_2$
norm. When the argument is a matrix, this denotes the associated spectral norm.}  
It is shown in \cite{Galeotti2017TargetingII}
that the optimal budget-constrained intervention for any $C$ can
be computed using the eigenvectors of $A$, and that in the
large-budget limit as $C$ tends to infinity, the optimal
intervention approaches $\sqrt{C} \cdot v_1(A)$.  Throughout, we
assume $v_i(A) $ is the unit $\ell_2$-norm eigenvector associated
with $\lambda_i$, the $i$th largest eigenvalue of a matrix $A$. We
also define $\alpha_i = \frac{1}{(1-\beta\lambda_i)^2}$, which is
the square of the corresponding eigenvalue of $M^{-1}$.  Note that
we do not consider eigenvalues to be ordered by absolute value;
this is done to preserve the ordering correspondence between
eigenvalues of $A$ and $M^{-1}$. $A$ may have negative
eigenvalues, but all eigenvalues of $M^{-1}$ will be positive when
$\beta \lambda_1 < 1$, as we will ensure throughout. 

The key result we use from \cite{Galeotti2017TargetingII} states that when $\beta$ is positive, 
as the budget increases
the cosine similarity between the optimal intervention $y^*$ and the first eigenvector of a graph, which we denote by $\rho(v_1(A), y^*)$,\footnote{
The cosine similarity of two non-zero vectors $z$ and $y$ is $\rho(z, y) = \frac{z \cdot y}{\norm{z}\norm{y}}$.
For unit vectors $x, y$, by the law of cosines,
    $\norm{x - y}^2 =\; 2(1 - \rho(x, y))$,
and so
    $1 - \frac{\norm{x-y}^2}{2} = \rho(x, y)$. 
Thus $\norm{x-y} < \epsilon$ for $\epsilon > 0$ if and only if $\rho(x, y) > 1 - \epsilon^2/2$.} approaches 1 at a 
rate 
depending on the (inverted) spectral gap of the adjacency matrix.\footnote{
It will sometimes be convenient for us to work with what we call the {\it inverted spectral gap} of a matrix $A$, which is the smallest value $\kappa$ such that $\abs{\lambda_i(A) \leq \kappa \cdot \lambda_1(A)}$.
}


Our results will involve quantifying the {\it competitive ratio} of an intervention $y$, which we define as
$\frac{W(y)}{W(y^*)}$, where $W(\cdot)$ denotes the social welfare at
equilibrium after an intervention vector is applied, and where $y^* = \arg\max_{x
  ~:~ \norm{x} = \sqrt C} W(x)$.
This ratio is at most 1, and maximizing it will be our objective for evaluating interventions.

\subsection{Random Graph Models}

We introduce several families of random graph distributions which we consider throughout. All of these models generate graphs which are undirected and have edges which are drawn independently. 

\begin{definition}[Random Graphs with Independent Edges]
A distribution of random graphs with independent edges is specified by a symmetric matrix $\ba \in [0,1]^{n \times n}$.
A graph is sampled by including each edge $(i,j)$ independently with probability $\ba_{ij}$. 
\end{definition}

Graphs with given expected degrees ($\gw$, or Chung-Lu graphs) and stochastic block model graphs, 
often used as models of realistic ``well-behaved'' networks,
are notable cases of this model which we will additionally focus on.


\begin{definition}[$G(w)$ Graphs]
  A $\gw$ graph is an undirected graph with an expected degree
  sequence given by a vector $\w$, whose length (which we denote by
  $n$) defines the number of vertices in the graph. For each pair of
  vertices $i$ and $j$ with respective expected degrees $w_i$ and
  $w_j$, the edge $(i, j)$ is included independently with probability
  $\ba_{ij} = \frac{w_i w_j}{\sum_{k \in [n]} w_k}$.
\end{definition}
Without loss of generality, we impose an ordering on $w_i$ values so
that $w_1\geq w_2\geq\ldots\geq w_n$. 
To ensure that each edge probability as described above is in $[0,1]$,
we assume throughout that for all vectors $\w$ we have 
that $w_1 \leq \sqrt{\sum_{k \in [n]} w_k}$.

$\gnp$ graphs and power law graphs are well-studied examples of graphs which can be generated by the $\gw$
model.\footnote{There are several other well-studied models of graphs
  with power law degree sequences, such as the BA preferential
  attachment model, as well as the fixed-degree model involving a
  random matching of ``half-edges''. Like the $\gw$ model, the latter
  model can support arbitrary degree sequences. We restrict ourselves
  to the independent edge model described above.} 
For $\gnp$ graphs, $w$ is a uniform vector where $w_i = np$ for each $i$.  Power law graphs are another notable special case where 
$w$ 
is a {\it power law sequence} $\left\lbrace w_i\right\rbrace_{i=1}^n$ such that $w_i = c\brack{i+i_0}^{-\frac{1}{\sigma-1}}$ for $\sigma > 2$, some constant $c>0$, and some integer $i_0\geq 0$. In such a sequence, the number of
elements with value $x$ is asymptotically proportional to $\frac{1}{x^{\sigma}}$.
  
\begin{definition}[Stochastic Block Model Graphs]
    A stochastic block model graph with $n$ vertices is undirected and has $m$ groups for some $m \leq n$. Edges are drawn independently according to a matrix $\ba$, and the probability of an edge between two agents depends only on their group membership. For any two groups $i$ and $j$, there is an edge probability $p_{ij}\in [0,1]$ such that $\ba_{kl} = p_{ij}$ for any agent $k$ in group $i$ and agent $l$ in group $j$.  
\footnote{If $m=n$, the stochastic block model can express any distribution of random graphs with independent edges, but will be most interesting when there are few groups.
}
\end{definition}  

For each graph model, one can choose to disallow self-loops by setting $\ba_{ii}=0$ for $1\leq i\leq n$, as is standard for $\gnp$ graphs. Our results will apply to both cases.

\section{Approximately Optimal Interventions}
\label{approxopt}

The main idea behind all of our intervention strategies is to target
efforts proportionally to the first eigenvector of the expected
adjacency matrix.  Here
we assume that this
eigenvector is known exactly. In \Cref{sbm}, we discuss cases when an
approximation of the eigenvector can be computed with zero or minimal
information about the graph.
Our main theorem for random graphs with independent edges shows conditions under which an intervention
proportional to the first eigenvector of the expected matrix $\ba$ is near-optimal.

We define a property for random graphs which we call {\it $(\epsilon, \delta)$-concentration} which will ensure that the expected first eigenvector constitutes a near-optimal intervention.
In essence, this is an explicit quantification of the asymptotic properties of ``large enough eigenvalues'' and ``non-vanishing spectral gap'' for {\it sequences} of random graphs from \cite{RePEc:arx:papers:1709.10402}. Intuitively, this captures graphs which are ``well-connected'' and not too sparse.
One can think of the first eigenvalue as a proxy for density, and the (inverse) second eigenvalue as a proxy for regularity or degree of clustering (it is closely related to a graph's mixing time). Both are important in ensuring concentration, and they trade off with each other (via the spectral gap condition) for any fixed values of $\epsilon$ and $\delta$.

\begin{definition}[$(\epsilon,\delta)$-Concentration]
A random graph with independent edges specified by $\ba$ satisfies $(\epsilon, \delta)$-concentration for $\epsilon, \delta\in(0,1)$ if:
\begin{enumerate}
    \item\label{dmax-rg-cond} The largest expected degree $d_{\max} = \max_i \sum_{j \in [n]} \ba_{ij}$
    is at least $\frac{4}{9}~\log(2n / \delta)$
    \item\label{li-rg-cond} The inverted spectral gap of $\ba$ is at most $\kappa$
    \item\label{l1-rg-cond} The quantity $\lambda_1(\ba) \cdot(1 - \kappa^2) $ is at least $\frac{1024\sqrt{d_{\max} \log(2n / \delta)}}{\epsilon^2}$
\end{enumerate}  
\end{definition}

\begin{theorem}\label{thm:sbm-main}

If $\ba$ satisfies $(\epsilon,\delta)$-concentration,
then with probability at least $1 - \delta$, 
the competitive ratio of $y = \sqrt{C}v_1(\ba)$ for a graph drawn from $\ba$ is at least $1 - \epsilon$ for a sufficiently large budget $C$ if
the spectral radius of the sampled matrix $A$ is less than $1/\beta$. 
\end{theorem}

The concentration conditions are used to show that the relevant spectral properties of generated graphs are almost surely close to their expectations, and the constraint on $\beta$ 
is necessary to ensure that actions and utilities are finite at equilibrium.\footnote{
The spectral radius condition holds with probability $1 - \delta$ when ${1}/{\beta}$ is at least ${\lambda_1(\ba) + \sqrt{4 d_{\max} \log(2n / \delta)}}$ (follows from e.g.\ \cite{chungradcliffe}, see \Cref{sec:imports} for details).
}
The sufficient budget will depend on the size of the {\it spectral
  gap}
of $\ba$, as well as the standard marginal values.
For example, if $\lambda_1 > 2\lvert \lambda_i \rvert$ holds in the realized graph for all $i>1$, then a budget of $C = 256 \cdot \norm{b}^2 / (\epsilon \beta \lambda_1(\ba) )^2$ will suffice. Intuitively, a large $\beta$ would mean more correlation between neighbors' actions at equilibrium. A large $\lambda_1\brack{\ba}$ would mean a denser graph (more connections between agents) in expectation and a large $\epsilon$ would mean that the realized graph is more likely to be close to expectation. All of these conditions reduce the required budget because a small intervention gets magnified by agent interaction. Further, the smaller the magnitude of initial $b$, the easier it is to change its direction.

The proof of \Cref{thm:sbm-main} is deferred to \Cref{omittedproofs}.
At a high level, our results proceed by first showing that the first
eigenvector is almost surely close to $v_1(\ba)$, then showing that the spectral gap is almost surely large enough such
that the first eigenvector is close to the
optimal intervention for appropriate budgets.  A key lemma for
completing the proof shows that interventions which are close to the
optimal intervention in cosine similarity have a competitive ratio
close to 1.

\begin{lemma}
    \label{cosineutil}
    Let $b$ be the vector of standalone values,
    and assume that $C > \max( \norm{b}^2, 1) $.
    For any $y$ where $\norm{y}^2 = C$ and $\rho(y, y^*) > \gamma$ for some $\gamma$, the competitive ratio of $y$ is at least $1 - 4\sqrt{2(1 - \gamma)}$.
\end{lemma}

The main idea behind this lemma is a smoothness argument for the welfare function. When considering interventions as points on the sphere of radius $\sqrt{C}$, small changes to an intervention cannot change the resulting welfare by too much. 
This additionally implies that when a vector $y$ is close to
$y^*$, 
the exact utility of $y^*$ 
for some budget $C$ 
can be achieved by an intervention proportional to $y$ with a
budget $C'$ which is not much larger than $C$. 

In \Cref{gw}, we give a specialization of \Cref{thm:sbm-main} to the case of $\gw$ graphs. There, the expected first eigenvector is proportional to $w$ when self-loops are not removed.
We give more explicit characterizations of the properties for $\gw$, $\gnp$, and power law graphs which ensure the above spectral conditions (i.e.\ without relying on eigenvalues), as well as a budget threshold for near-optimality.  
We discuss the steps of the proof in greater detail, and they are largely symmetric to the steps required to prove \Cref{thm:sbm-main}.

\section{Centrality Estimation} \label{sbm}


The previous sections show that interventions proportional to $v_1(\ba)$ are often near-optimal simultaneously for almost all graphs generated by $\ba$.
While we often may have domain knowledge about a network which helps characterize its edge distribution, we still may not be able to precisely approximate the first eigenvector of $\ba$ {\it a priori}.   
In particular, even if we believe our graph comes from a power law distribution, we may be at a loss in knowing which vertices have which expected degrees.

In this section, we discuss approaches for obtaining near-optimal interventions without initial knowledge of $\ba$.
We first observe that ``blind'' interventions, which treat all vertices equally in expectation, will fail to approach optimality. We then consider statistical estimation techniques for approximating the first eigenvector which leverage the special structure of $\gw$ and stochastic block model graphs.
In each case, we identify a simple {\it target intervention}, computable directly from the realized graph, 
which is near-optimal when $(\epsilon,\delta)$-concentration is satisfied. We then give efficient sampling methods for approximating these target interventions.
Throughout \Cref{sbm}, our focus is to give a broad overview of these techniques 
rather than to present them as concrete algorithms, and we frequently omit constant-factor terms with asymptotic notation.



\subsection{Suboptimality of Blind Interventions}

Here we begin by showing that when the spectral gap is large, all interventions which are far from the optimal intervention in cosine similarity will fail to be near-optimal even if the budget is very large. 
\begin{lemma}\label{lemma:util-upper-bound}
  Assume
  that $C$ is sufficiently large such that the role of standalone values is negligible. For any $y$ where $\norm{y}^2 = C$
  and $\rho(y, y^*) < \gamma$, the competitive ratio is bounded by
\begin{align*}
  \gamma^2\brack{1-\frac{\alpha_2}{\alpha_1}} +  \frac{\alpha_2}{\alpha_1} + 2\sqrt{\frac{\alpha_2}{\alpha_1}}, 
\end{align*}
where $\alpha_i$ is the square of the $i$th largest eigenvalue of $M^{-1}$.
\end{lemma}
This tells us that if one were to design an intervention without using
any information about the underlying graph,
the intervention is
unlikely to do well compared to the optimal one for the same budget unless eigenvector centrality is uniform, as in the case of $\gnp$ graphs. Thus, there is a need to
try to learn graph information to design a close-to-optimal
intervention. We discuss methods for this next.

\subsection{Degree Estimation in $\gw$ Graphs}
\label{sec:gw-est}
For $\gw$ graphs, we have seen that expected degrees suffice for near-optimal interventions, and we show that  degrees can suffice as well. 

\begin{lemma}
\label{lemma:real-deg-near-opt}

    If a $\gw$ graph specified by $\ba$ satisfies $(\epsilon, \delta)$-concentration, 
    then with probability at least $1 - O(\delta)$,
\begin{align*}
\norm{w- {w^*}} \leq&\; O(\epsilon \norm{w}),
\end{align*}
where $w^*$ is the empirical degree vector, and
the intervention proportional to $w^*$ obtains a competitive ratio of $1 - O(\epsilon)$ when the other conditions for \Cref{thm:sbm-main} are satisfied.
\end{lemma}

Thus, degree estimation is our primary objective in considering statistical approaches. As we can see from the analysis in \Cref{thm:sbm-main,thm:gw-main}, if we can estimate the unit-normalized degree vector $w^*$ to within $\epsilon$ $\ell_2$-distance, our competitive ratio for the corresponding proportional intervention will be $1 - O(\epsilon)$.
Our approaches focus on different query models, respresenting the types of questions we are allowed to ask about the graph; these query models are also studied for the problem of estimating the average degree in a graph \cite{goldreich,10.1145/2566486.2568019}. If we are allowed to query agents' degrees, near-optimality follows directly from the above lemma, so we consider more limited models.

\subsubsection*{Edge Queries.}
Suppose we are allowed to query whether an edge exists between two vertices. We can then reduce the task of degree estimation to the problem of estimating the mean of $n$ biased coins, where for each vertex, we ``flip'' the corresponding coin by picking another vertex uniformly at random to query. By Hoeffding and union bounds, $O\brack{\frac{n}{\epsilon^2} \log\brack{\frac{n}{\delta}} }$ total queries suffice to ensure that 
with probability $1 - \delta$,  
each degree estimate is within $\epsilon n $ additive error. Particularly in the case of dense graphs,
and when $\epsilon$ is not too small compared to $1/n$, this will be considerably more efficient than reconstructing the entire adjacency matrix. In particular, if $\norm{w}_1 = \Theta(n^2)$, the above error bound on additive error for each degree estimate directly implies that the estimated degree vector $\hat{w}$ is within  $\ell_1$ (and thus $\ell_2$) distance of $O(\epsilon \norm{w}_2)$.

\subsubsection*{Random Neighbor Queries.}
Suppose instead we are restricted to queries which give us a uniformly random neighbor of a vertex. 
We give an approach wherein queries are used to conduct a random walk in the graph. 
The stationary distribution is equivalent to the the first eigenvector of the {\it diffusion matrix} $P = AD^{-1}$, where $D$ is the diagonal matrix of degree counts.\footnote{
The stationary distribution of a random walk on a simple connected
graph is $\frac{d_i}{\sum_j d_j}$ for all vertices $i$, where $d_i$ is the degree. While $\gw$ graphs may fail to be connected, in many cases the vast majority of vertices will belong to a single component, and we can focus exclusively on that component.
We show this in \Cref{lemma:gw-mixing-time}.
}
We can then learn estimates of degree proportions by sampling from the stationary distribution via a random walk.

The mixing time of a random walk on a graph determines the number of steps required such that the probability distribution over states is close to the stationary distribution in total variation distance. 
We can see that for $\gw$ graphs satisfying $(\epsilon,\delta)$-concentration with a large enough minimum degree, 
mixing times will indeed be fast.
\begin{lemma}
    \label{lemma:gw-mixing-time}
For $\gw$ graphs satisfying $(\epsilon, \delta)$-concentration and with $w_n \geq \frac{1}{\epsilon}$, the mixing time of a random walk to within $\epsilon$ total variation distance to the stationary distribution is $O(\log(n/\epsilon))$. Further, the largest connected component in $A$ contains $n(1 - \exp\parens{-O(1/\epsilon)}$ vertices in expectation. 
\end{lemma}

If a random walk on our graph has some mixing time $t$ to an approximation of the stationary distribution, we can simply record our location after every $t$ steps to generate a sample.
Using standard results on learning discrete distributions (see e.g. \cite{canonne2020short}),
$O\parens{\frac{n + \log(1 / \delta)}{\epsilon^2}}$ samples from $\epsilon$-approximate stationary distributions suffice to approximate $w^*$ within $\ell_1$ distance of $O(\epsilon \norm{w^*})$ with probability $1- \delta$, directly giving us the desired $\ell_2$ bound. Joining this with \Cref{lemma:gw-mixing-time}, our random walk takes a total of $O\brack{\frac{n+ \log(1/\delta)}{\epsilon^2} \log\brack{\frac{n}{\epsilon}}}$ steps (and thus queries) to obtain our target intervention, starting from an arbitrary vertex in the largest connected component.

\subsection{Matrix Reconstruction in SBM Graphs}
\label{subsec:sbm-estimation}
There is a fair amount of literature on estimation techniques for stochastic block model graphs, often focused on cases where group membership is unknown \cite{NIPS2019_9498,JMLR:v18:16-480,tabouy2017variational,schaub2019blind}.  
The estimation of eigenvectors is discussed in
\cite{Medina2018}, where they consider stochastic block model graphs as a
limit of a convergent sequence of ``graphons''.
Our interest is primarily in recovering eigenvector centrality efficiently from sampling, and we will make the simplifying assumption that group labels are visible for all vertices. This is reasonable in many cases where a close proxy of one's primary group identifier (e.g.\ location, job, field of study) is visible but connections are harder to map.

In contrast to the $\gw$ case, degree estimates no longer suffice 
for estimating the first eigenvector. 
We assume that there are $m$ groups and that we know each agent's group. Our aim will be to estimate the relative densities of connection between groups. 
When there are not too many groups, the parameters of a stochastic block model graph can be estimated efficiently with either edge queries or random neighbor queries, 
From here, we can construct an approximation of $\ba$ and compute its first eigenvector directly. In many cases, the corresponding intervention is near-optimal.

A key lemma in our analysis shows that the ``empirical block matrix'' is close to its expectation in spectral norm. We prove this for the case where all groups are of similar sizes, but the approach can be generalized to cover any partition.

\begin{lemma}
\label{lemma:sbm-block-bound}
For a stochastic block model graph generated by $\ba$ with $m$ groups, each of size $O(\frac{n}{m})$, let $\hat{A}$ denote the empirical block matrix of edge frequencies for each group. Each entry per block in $\hat{A}$ will contain the number of edges in that block divided by the size of the block. With probability at least $1 - \delta$,
\begin{align*}
    \norm{\ba - \hat{A}} \leq&\; 
O\left(\max{\left(\frac{m\sqrt{\log(n/\delta)}}{\sqrt{n}}, ~\log^2(n/\delta)\right)}\right).
\end{align*}
\end{lemma}

The same bound will then apply to the difference of the first eigenvectors, rescaled by the first eigenvalues (which will also be close). 
Similar bounds can also be obtained when group sizes may vary, but we stick to this case for simplicity.

\subsubsection*{Edge Queries.}

If we are allowed to use edge queries, we can estimate the empirical edge frequency for each of the $O(m^2)$ pairs of groups by repeatedly sampling a vertex uniformly from each group and querying for an edge. 
This allows reconstruction of the empirical frequencies up to $\epsilon$ error for each group pair, with probability $1 - \delta$, with $O\left(\frac{m^2}{\epsilon^2} \log(m/\delta) \right)$ samples. 
For the block matrix $\hat{A}$ of edge frequencies for all group pairs, \Cref{lemma:sbm-block-bound} implies that this will be close to its expectation when there are not too many groups, and so our estimation will be close to $\ba$ in spectral norm as well. 
If $\ba$ satisfies $(\epsilon,\delta)$-concentration
and the bound from \Cref{lemma:sbm-block-bound} is small compared to the norm of $\ba$, then the first eigenvectors of $A$, $\ba$, and $\hat{A}$ will all be close, and the corresponding intervention proportional to $v_1(\hat{A})$ will be near-optimal. 

When all group pairs may have unique probabilities, this will only provide an advantage over a naive graph reconstruction with $O(m^2)$ queries in the case where $m = o(n)$. 
If we know that all out-group probabilities are the same across groups, our dependence on $m$ becomes linear, as we can treat all pairs of distinct groups as one large group. If in-group probabilities are the same across groups as well, the dependence on $m$ vanishes, as we only have two probabilities to estimate.

\subsubsection*{Random Neighbor Queries.}
We can also estimate the empirical group frequency matrix with random neighbor queries.
For each group, the row in $\ba$ corresponding to the edge probabilities with other groups can be interpreted as a distribution of frequencies for each group.
$O(\frac{m}{\epsilon^2} \log(\frac{m}{\delta}))$ samples per row suffice to get additive error at most $\epsilon$ for all of the 
relative connection probabilities for our chosen group.
This lets us estimate each of the $m$ rows up to scaling, at which point we can use the symmetry of the matrix to recover an estimate of $\ba$ up to scaling by some factor. 
Again, when $(\epsilon, \delta)$-concentration holds and the bound from \Cref{lemma:sbm-block-bound} is small,
the first eigenvector of this estimated matrix will give us a near-optimal intervention.

\section{Experiments}\label{experiments}

Our theoretical results require graphs to be relatively large in order for the obtained bounds to be nontrivial.   
It is natural to ask how well the heuristic interventions we describe will perform on relatively small random graphs, as well as on real-world graphs which do not come from a simple generative model (and may not have independent edges).
Here, we evaluate our described interventions on real and synthetic network data, by adjusting $b_i$ values and computing the resulting welfare at equilibrium, and find that performance can be quite good even on small graphs.
Our experimental results on synthetic networks are deferred to \Cref{sec:gnp_exp}.

\subsection{Real Networks}

To test the usefulness of our results for real-world networks which we expect to be ``well-behaved'' according to our requirements,
we simulate the intervention process using network data collected from villages in South India, for purposes of targeted microfinance deployments,
from \cite{Banerjee1236498}. In this context, we can view actions $a_i$ as indicating levels of economic activity, which we wish to stimulate by increasing individual propensities for spending and creating network effects.
The dataset contains 
many graphs for each village using different edge sets (each representing different kinds of social connections), as well as graphs where nodes are households rather than individuals. We use the household dataset containing the union of all edge sets. These graphs have degree counts ranging from 77 to 365, and our experiments are averaged over 20 graphs from this dataset. We plot competitive ratios while varying $C$ (scaled by network size) and the spectral radius of $\beta A$, fixing $b_i = 1$ for each agent. 

The expected degree intervention is replaced by an intervention proportional to exact degree. We also fit a stochastic block model to graphs using a version of the approach described in \Cref{subsec:sbm-estimation}, using exact connectivity probabilities rather than sampling. Our group labels are obtained by running the Girvan-Newman clustering algorithm \cite{Girvan7821} on the graph, pruning edges until there are either at least 10 clusters with 5 or more vertices or 50 clusters total.
We evaluate the intervention proportional to the first eigenvector of the reconstructed block matrix. All interventions are compared to a baseline, where no change is applied to $b$, for demonstrating the relative degree in social welfare change.

\begin{figure}[ht]
    \begin{center}
    \includegraphics[width=5.5cm]{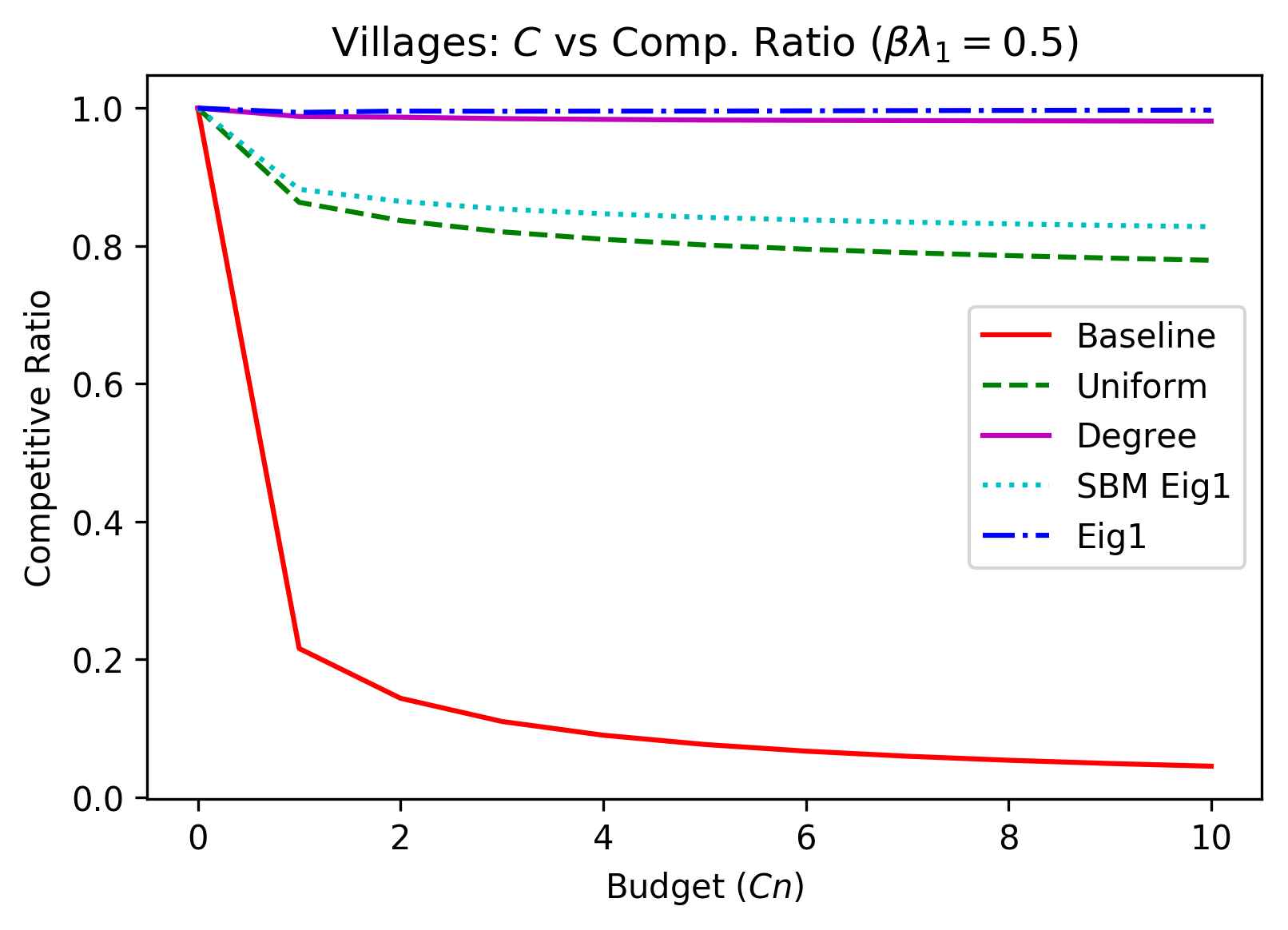}
    \includegraphics[width=5.5cm]{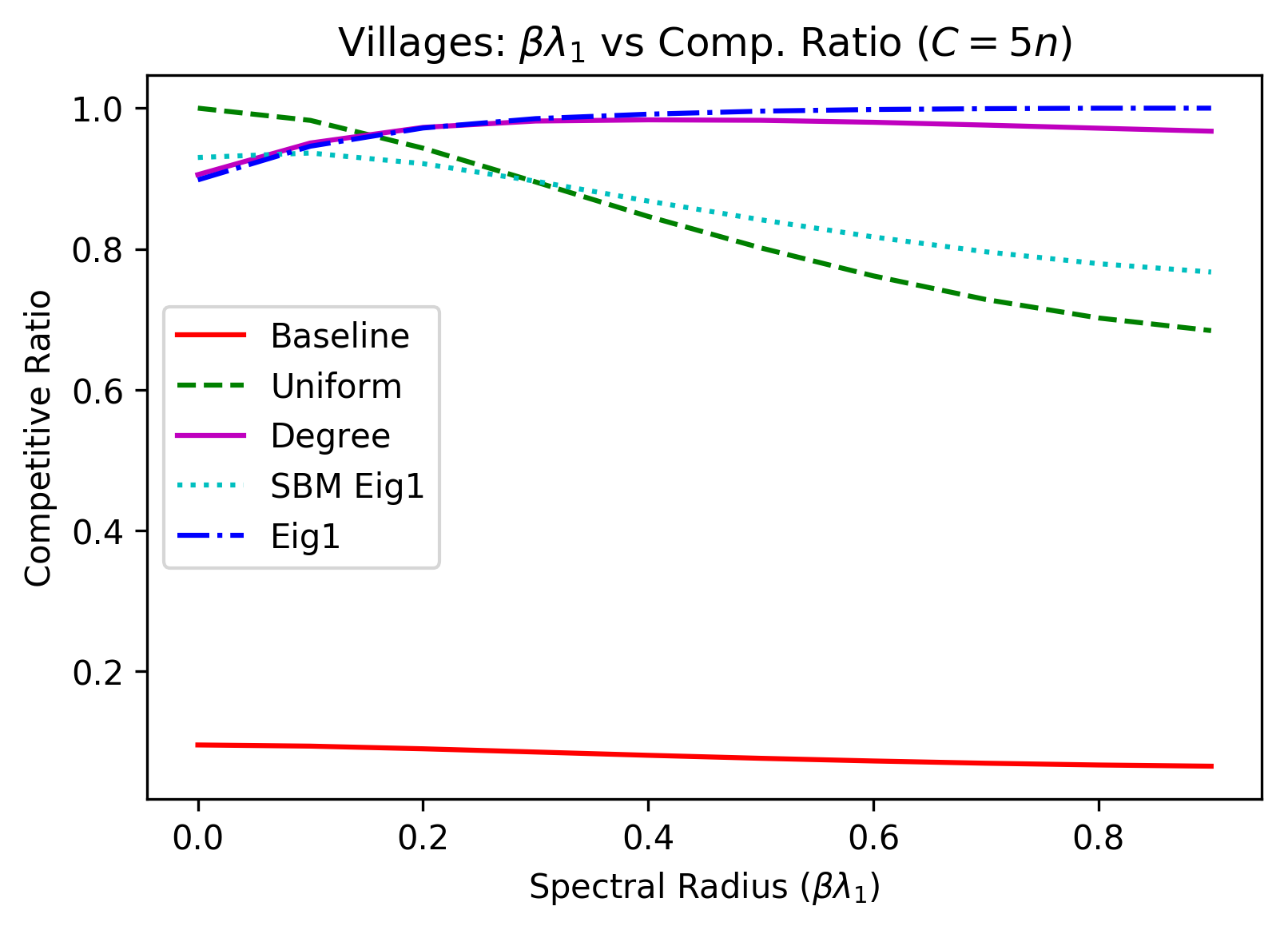}
    \end{center}
    \caption{Intervention in Village Graphs}
    \label{village-plots}
    \centering
\end{figure}

In \Cref{village-plots}, we find that degree interventions perform quite well, and are only slightly surpassed by first eigenvector interventions. 
The stochastic block model approach performs better than uniform when the spectral radius sufficiently large, but is still outperformed by the degree and first eigenvector interventions. 
Upon inspection, the end result of the stochastic block model intervention was often uniform across a large subgraph, with little or no targeting for other vertices, 
which may be an artifact of the clustering method used for group assignment.
On the whole, we observe that minimal-information approaches can indeed perform quite well on both real and simulated networks.

\iffull
\subsubsection*{Acknowledgments.}
We thank Ben Golub, Yash Kanoria, Tim Roughgarden, Christos Papadimitriou, and anonymous reviewers for their invaluable feedback. 
\fi

\bibliographystyle{splncs04}
\bibliography{ref}

\begin{thebibliography}{10}
\providecommand{\url}[1]{\texttt{#1}}
\providecommand{\urlprefix}{URL }
\providecommand{\doi}[1]{https://doi.org/#1}

\bibitem{JMLR:v18:16-480}
Abbe, E.: Community detection and stochastic block models: Recent developments.
  Journal of Machine Learning Research  \textbf{18}(177),  1--86 (2018)

\bibitem{10.1145/335305.335326}
Aiello, W., Chung, F., Lu, L.: A random graph model for massive graphs. In:
  Proceedings of the Thirty-Second Annual ACM Symposium on Theory of Computing.
  p. 171–180. STOC ’00, Association for Computing Machinery, New York, NY,
  USA (2000)

\bibitem{Akbarpour:2018:DSV:3219166.3219225}
Akbarpour, M., Malladi, S., Saberi, A.: Diffusion, seeding, and the value of
  network information. In: Proceedings of the 2018 ACM Conference on Economics
  and Computation. pp. 641--641. EC '18, ACM, New York, NY, USA (2018)

\bibitem{DBLP:journals/corr/Avella-MedinaPS17}
Avella{-}Medina, M., Parise, F., Schaub, M.T., Segarra, S.: Centrality measures
  for graphons. CoRR  \textbf{abs/1707.09350} (2017),
  \url{http://arxiv.org/abs/1707.09350}

\bibitem{Medina2018}
Avella{-}Medina, M., Parise, F., Schaub, M.T., Segarra, S.: Centrality measures
  for graphons (2017), \url{http://arxiv.org/abs/1707.09350}

\bibitem{Banerjee1236498}
Banerjee, A., Chandrasekhar, A.G., Duflo, E., Jackson, M.O.: The diffusion of
  microfinance. Science  \textbf{341}(6144),  1236498 (2013).
  \doi{10.1126/science.1236498}

\bibitem{10.1093/restud/rdz008}
Banerjee, A., Chandrasekhar, A.G., Duflo, E., Jackson, M.O.: {Using Gossips to
  Spread Information: Theory and Evidence from Two Randomized Controlled
  Trials}. The Review of Economic Studies  \textbf{86}(6),  2453--2490 (02
  2019)

\bibitem{banerjee14SSRN}
Banerjee, A., G~a, A., Duflo, E., Jackson, M.: Gossip: Identifying central
  individuals in a social network (06 2014)

\bibitem{Barabasi509}
Barab{\'a}si, A.L., Albert, R.: Emergence of scaling in random networks.
  Science  \textbf{286}(5439),  509--512 (1999)

\bibitem{Breza2017}
Breza, E., Chandrasekhar, A.G., McCormick, T.H., Pan, M.: Using aggregated
  relational data to feasibly identify network structure without network data
  (2017)

\bibitem{Candogan2012}
Candogan, O., Bimpikis, K., Ozdaglar, A.: Optimal pricing in networks with
  externalities. Operations Research  \textbf{60}(4),  883--905 (2012)

\bibitem{canonne2020short}
Canonne, C.L.: A short note on learning discrete distributions (2020)

\bibitem{Chung2002}
Chung, F., Lu, L.: Connected components in random graphs with given expected
  degree sequences. Annals of Combinatorics  \textbf{6}(2),  125--145 (Nov
  2002)

\bibitem{Chung2003}
Chung, F., Lu, L., Vu, V.: Eigenvalues of random power law graphs. Annals of
  Combinatorics  \textbf{7}(1),  21--33 (Jun 2003)

\bibitem{chungradcliffe}
Chung, F., Radcliffe, M.: On the spectra of general random graphs. Electr. J.
  Comb.  \textbf{18} (10 2011)

\bibitem{RePEc:arx:papers:1709.10402}
Dasaratha, K.: {Distributions of Centrality on Networks}. Papers, arXiv.org
  (Sep 2017)

\bibitem{10.1145/2566486.2568019}
Dasgupta, A., Kumar, R., Sarlos, T.: On estimating the average degree. In:
  Proceedings of the 23rd International Conference on World Wide Web. p.
  795–806. WWW ’14, Association for Computing Machinery, New York, NY, USA
  (2014)

\bibitem{Demange2017}
Demange, G.: Optimal targeting strategies in a network under complementarities.
  Games and Economic Behavior  \textbf{105},  84--103 (Sep 2017)

\bibitem{eckles2019seeding}
Eckles, D., Esfandiari, H., Mossel, E., Rahimian, M.A.: Seeding with costly
  network information (2019)

\bibitem{Galeotti2017TargetingII}
Galeotti, A., Golub, B., Goyal, S.: Targeting interventions in networks (2017)

\bibitem{Girvan7821}
Girvan, M., Newman, M.E.J.: Community structure in social and biological
  networks. Proceedings of the National Academy of Sciences  \textbf{99}(12),
  7821--7826 (2002). \doi{10.1073/pnas.122653799}

\bibitem{goldreich}
Goldreich, O., Ron, D.: Approximating average parameters of graphs. Random
  Structures \& Algorithms  \textbf{32}(4),  473--493 (2008)

\bibitem{Huang2019TheVO}
Huang, J., Mani, A., Wang, Z.: The value of price discrimination in large
  random networks. Proceedings of the 2019 ACM Conference on Economics and
  Computation  (2019)

\bibitem{Kleinberg1999TheWA}
Kleinberg, J.M., Kumar, R., Raghavan, P., Rajagopalan, S., Tomkins, A.S.: The
  web as a graph: Measurements, models, and methods. In: Asano, T., Imai, H.,
  Lee, D.T., Nakano, S.i., Tokuyama, T. (eds.) Computing and Combinatorics. pp.
  1--17. Springer Berlin Heidelberg, Berlin, Heidelberg (1999)

\bibitem{10.1145/1150402.1150479}
Leskovec, J., Faloutsos, C.: Sampling from large graphs. In: Proceedings of the
  12th ACM SIGKDD International Conference on Knowledge Discovery and Data
  Mining. p. 631–636. KDD ’06, Association for Computing Machinery, New
  York, NY, USA (2006)

\bibitem{parise2018graphon}
Parise, F., Ozdaglar, A.: Graphon games: A statistical framework for network
  games and interventions (2018)

\bibitem{ruggeri2019sampling}
Ruggeri, N., Bacco, C.D.: Sampling on networks: estimating eigenvector
  centrality on incomplete graphs (2019)

\bibitem{schaub2019blind}
Schaub, M.T., Segarra, S., Tsitsiklis, J.N.: Blind identification of stochastic
  block models from dynamical observations (2019)

\bibitem{10.1007/BFb0023849}
Sinclair, A.: Improved bounds for mixing rates of markov chains and
  multicommodity flow. In: Simon, I. (ed.) LATIN '92. pp. 474--487. Springer
  Berlin Heidelberg, Berlin, Heidelberg (1992)

\bibitem{tabouy2017variational}
Tabouy, T., Barbillon, P., Chiquet, J.: Variational inference for stochastic
  block models from sampled data (2017)

\bibitem{NIPS2019_9498}
Yun, S.Y., Proutiere, A.: Optimal sampling and clustering in the stochastic
  block model. In: Wallach, H., Larochelle, H., Beygelzimer, A.,
  d\textquotesingle Alch\'{e}-Buc, F., Fox, E., Garnett, R. (eds.) Advances in
  Neural Information Processing Systems 32, pp. 13422--13430. Curran
  Associates, Inc. (2019)

\end{thebibliography}

\clearpage
\appendix

\section{Graphs with Given Expected Degrees}
\label{gw}


In this section, we show a method for obtaining near-optimal
interventions in graphs generated by the $\gw$ model.  We show that
the first eigenvector of a $\gw$ graph is almost surely close to $\w$,
and our intervention will simply be proportional to $\w$.
This indicates that degree estimates are often sufficient for
near-optimal intervention. 
We assume that $\w$ is sorted in
descending order and that each entry is strictly positive.  
Our main theorem for this section holds for all $\gw$ distributions which satifsy the following specialization of $(\epsilon,\delta)$-concentration. The second condition corresponds to requiring a sufficiently large first eigenvalue, as was the case for general random graphs; $\gw$ graphs do not exhibit clustering on average, and so we do not need an additional condition for the second eigenvalue.


\begin{definition}[$(\epsilon,\delta)$-Concentration for $\gw$ Graphs]
A $\gw$ graph satisfies $(\epsilon, \delta)$-concentration for $\epsilon, \delta \in (0,1)$ if:
\begin{enumerate}
    \item The largest expected degree $\w_1$ is at least $\frac{4}{9} \log(2n / \delta)$ and at most $\frac{\norm{\w}}{6}$   
    \item The second-order average of the expected degree sequence $\tilde{d} =  \frac{\sum_i \w^2_i}{\sum_i w_i}$ is at least    
    $\frac{256 \brack{ \sqrt{4w_1 \log(2n/\delta)} + 1}  }{\epsilon^2}$
\end{enumerate}
\end{definition}

\begin{theorem} \label{thm:gw-main}
    For $\gw$ distributions satisfying $(\epsilon, \delta)$-concentration, and for $C$ at least $\frac{256 \norm{b}^2}{(\epsilon \beta \tilde{d})^2} $, 
    with probability at least $1 - \delta$, 
    \begin{align*}
        \frac{W\left(\sqrt{C}\cdot \frac{\w}{\norm{w}} \right)}{W\left(y^* \right)} \geq 1 - \epsilon,
    \end{align*}
    where $b$ is the vector of standalone marginal values and $y^*$ is the optimal intervention for a budget $C$, if the spectral radius of $A$ is less than $1/\beta$.\footnote{
The spectral radius of $\beta A$ is less than 1 with probability $1- \delta$ when $\frac{1}{\beta} < \frac{1}{\tilde{d} + \sqrt{4 w_1 \log(2n/\delta)}}.$
}
\end{theorem}

Our conditions ensure that the first eigenvalue of the
graph is not too small, and that the other eigenvalues are not too
large in magnitude.
It is worth noting that $\frac{1}{\beta \tilde{d}}$ is small unless the network effects in the game are negligible, in which case we should not expect eigenvector centrality to be important for small budgets.  
In \Cref{gnp}, we consider applications to $\gnp$ and power law
graphs.
Here we assume the vector $\w$ which parameterizes the graph distribution is known, and in fact our intervention will simply be proportional to $\w$. 

In the proof of Theorem \ref{thm:gw-main}, we proceed by observing that the first eigenvector of $\ban$ is proportional to $\w$, and then prove that when the spectral conditions hold, the first eigenvector of $\an$ is nearly proportional to $\w$ with high probability. We then determine sufficient budget sizes such that near-optimality of the intervention follows from \Cref{cosineutil}.

\subsection{Proportionality of Eigenvector Centrality and Degree}
\label{subsec:proportionality}

Here we show that the first eigenvector of the adjacency matrix $A$ is almost surely close to the unit vector rescaling of $w$. 
We first observe that this holds in the standard version of the $\gw$ model which allows for self-loops, and then we show that the first eigenvector does not change by much upon pruning loops.  
\begin{lemma}
    \label{lemma:eigen-proportionality}
    For $\gw$ graph distributions, any eigenvector of $\ba$ corresponding to a non-zero eigenvalue is proportional to $\w$.
\end{lemma}

It can be checked that up to scaling, $\w$ is the unique vector which satisfies the eigenvalue equation for a non-zero eigenvalue.
As all rows and columns of $\ba$ are proportional to $\w$, there is only one non-zero eigenvalue. 
From the eigenvalue equation it is simple to check that the non-zero eigenvalue $\lambda_1(\ba)$ will be equal to the second-order average degree $\frac{\sum_i w_i^2}{\sum_i w_i}$.

In this formulation of the $\gw$ model, agents are allowed to have self-loops with positive probability. 
We note now that even if we remove the possibility of self-loops by setting the diagonal entries of $\ba$ to 0, which in turn will remove the
rank-deficiency, the spectral norm between the expected matrix (with loops) and the realized matrix (with or without loops) will be small.
This in turn will imply that the first eigenvectors of these matrices are close.


A key tool in this proof is a bound on 
the difference in first eigenvectors of matrices which are close in norm when one of them has a small second eigenvalue. We only make use of this in the case where the second eigenvalue of one matrix is zero, but a more general version of the result, used to prove \Cref{thm:sbm-main}, is included in \Cref{omittedproofs}.

\begin{lemma} \label{lemma:simple-eigenvector-bound}
    Let $A$ be a symmetric $n \times n$ matrix with largest absolute eigenvalue $\lambda_1$ (with multiplicity 1) and all other eigenvalues equal to 0. Let $B$ be a symmetric $n \times n$ matrix with largest absolute eigenvalue $\mu_1$, and suppose $\norm{A - B} \leq \eta$.
Then,
\begin{align*}
    \norm{v_1(A) - v_1(B)} \leq&\; \sqrt{ 2\left(1 - \frac{\mu_1 - \eta}{\lambda_1}\right)}.
\end{align*}
\end{lemma}

This follows from considering a decomposition of the first eigenvector
of $B$ into a component proportional to $v_1(A)$ and one proportional
to some vector orthogonal to $v_1(A)$. Given that $A-B$ has a small
spectral norm, the image of $v_1(B)$ in $A$ will be large, showing
that the orthogonal component is small, which we can use to show that
the eigenvectors are close.

We can then show that norm difference of the expected and realized
matrices is almost surely small using a matrix concentration bound
from \cite{chungradcliffe}, allowing us to apply
\Cref{lemma:simple-eigenvector-bound} to bound the difference in their
eigenvectors, as the second eigenvalue of $\ba$ is 0.

\begin{lemma} \label{lemma:1st-eigenvector-close}
If the above assumptions about the $\gw$ distribution are satisfied, then with probability $1-\delta$ it holds that:
\begin{align*}
    \norm{v_1(A) - \frac{\w}{\norm{\w}}} \leq&\; \epsilon / 8.
\end{align*}
\end{lemma}

As the first eigenvector of $\ba$ is proportional to $\w$, this shows
that that the first eigenvector will be nearly proportional to $\w$
regardless of whether we allow self-loops, and so our intervention
will be close to the true graph's first eigenvector. Next, we will see
that this implies near-optimality for a sufficiently large budget.

\subsection{Bounding Suboptimality of Interventions}
\label{subsec:subopt}

The previous results indicate that the first eigenvector will be close
to its expectation when our assumptions hold, even upon removing
self-loops. 
We now give a similar results for the first and second eigenvalues,
which allows us to guarantee a sufficiently large spectral gap for $M^{-1}$. First,
we give a bound on the second eigenvalue of $\ba$ with the diagonal
removed.

\begin{lemma}\label{lemma:2nd-eigenvalue-bound}
    Let $D$ be the matrix which is equal to $\ba$ along the diagonal and 0 elsewhere, and let $\lambda_2$ denote the second-largest absolute eigenvalue of 
  $\ba-D$. 
  Then, 
  \begin{equation*}
    \lambda_2 \leq \frac{2w_1\lambda_1}{\norm w} + 1.
  \end{equation*}
\end{lemma}
This follows from a similar orthogonal decomposition approach to the proof of \Cref{lemma:simple-eigenvector-bound}, as well as another direct application of \Cref{lemma:simple-eigenvector-bound}.
We can then get an absolute bound on the first and second eigenvalues of $A$.

\begin{lemma}
    \label{lemma:2nd-eigenvalue-bound-2}
    With probability at least $1 - \delta$, the second eigenvalue of $A$ is at most
$$ \frac{2 w_1 \tilde{d}}{\norm{w}} + 1 + \sqrt{4 w_1 \log(2n / \delta)}  $$
and the first eigenvalue of $A$ is at least 
$$\tilde{d} - 1 - \sqrt{4 w_1 \log(2n / \delta) }.$$
\end{lemma}
This follows from applying the triangle inequality to \Cref{lemma:2nd-eigenvalue-bound} and 
Theorem 1 from \cite{chungradcliffe} as well as from our observation about the first expected eigenvalue of $\ba$. 
To complete the proof of \Cref{thm:gw-main}, we can combine the previous results to show that $\lambda_1 > 2 \lambda_2$ when the stated conditions hold. Proposition 2 from \cite{Galeotti2017TargetingII} allows us to use this fact to show that when budget is above our lower bound, the first eigenvector is close to the optimal intervention in cosine similarity. We can then show that $w$ and the optimal intervention are close in cosine similiarity, using \Cref{lemma:1st-eigenvector-close} as a key step. Plugging this into \Cref{cosineutil} gives us the theorem.

\subsection{Examples: 
$\gnp$ and Power Law Graphs}
\label{gnp}

$\gnp$ graphs are perhaps the most well-studied family of random graphs, which we can interpret as a special case of the $\gw$ model and give explicit conditions for when $(\epsilon, \delta)$-concentration holds. These conditions are lower bounds on $n$ and $p$ which guarantee the requirements for applying \Cref{thm:gw-main}, and for clarity we restate the eigenvector similarity and near-optimality results for the case of $\gnp$ graphs. 

\begin{lemma}
    For $\gnp$ graphs with $p \geq \frac{4 \log(2n/\delta)}{9(n-1)}$ and $n$ at least ${\Omega}(1 / \epsilon^4) $, 
    \begin{align*}
        \norm{v_1(A) - \frac{1}{\sqrt{n}}} \leq&\; \epsilon/8.
    \end{align*}
    \label{lemma:gnp-eigenvector-bound}
\end{lemma}

\begin{theorem}
    For $G(n, p)$ graphs with $p \geq \frac{4 \log (2n / \delta) }{9n-1}$ and  $n$ at least $\Omega (1/\epsilon^4)$,      
with probability at least $1- \delta$, 
the uniform intervention $y = \frac{\sqrt{C}}{\sqrt n} \cdot \mathbf{1}$ achieves utility within $1 - \epsilon$ of the optimal intervention for budgets $C$ at least $\frac{256 \norm{b}^2}{\left(\epsilon \beta np \right)^2}$.   
\label{thm:gnp-main}
\end{theorem}

The constant factor for the lower bound on $n$ we obtain in the proof of \Cref{lemma:gnp-eigenvector-bound} is large, but can likely be optimized, and our empirical results in \Cref{experiments} indicate the uniform intervention is close to optimal on reasonably small $\gnp$ graphs.

This theorem also applies directly to stochastic block model graphs where all blocks are the same size and each has the same ingroup and outgroup probabilities. Here, all agents are equally central in expectation, and it is simple to check that the first eigenvector of $\ba$ will be uniform. This allows a near-optimal intervention for sufficiently large and dense graphs without any knowledge of the edges, connectivity probabilities, group memberships, or degrees.

We can also show that our approach is near-optimal
for many power law graphs, as introduced in \Cref{prelims}. 
Power law graphs are a notable special case of the $G(w)$ model, and are often studied as models of real-world networks. Our results hold for power law graphs where $\sigma \in (2, 2.5)$, a range containing many observed examples \cite{Kleinberg1999TheWA,Barabasi509,10.1145/335305.335326}.

\begin{lemma}\label{lemma:power-law-eig}
For power law graphs with $2 < \sigma < 2.5$, 
if we have that $w_1$ is at least 
$\Omega\left( \left( \frac{\log(2n/\delta)}{\epsilon^4} \right)^{\frac{1}{5 - 2\sigma}} \right)$, 
then with probability at least $1 - \delta$, 
\begin{align*}
    \norm{v_1(A) - \frac{w}{\norm{w}}} \leq&\; \epsilon / 8.
\end{align*}
\end{lemma}

\begin{theorem}\label{thm:power-law-main}
For power law graphs with $2 < \sigma < 2.5$, 
if we have that $w_1$ is at least 
$\Omega\left( \left( \frac{\log(2n/\delta)}{\epsilon^4} \right)^{\frac{1}{5 - 2\sigma}} \right)$, 
if $\beta$ is small enough to ensure that $\beta \lambda_1(A) < 1$,
then with probaility at least $1 - \delta$, the intervention $y = \sqrt{C} \cdot \frac{\w}{\norm{\w}_1}$ achieves utility within a $1 - \epsilon$ factor of the optimal intervention for budgets $C$ at least $\frac{256 \norm{b}^2}{\left(\epsilon \beta \lambda_1(\ba) \right)^2}$, where $\lambda_1(\ba) = \Theta(w_1^{3 - \sigma})$.
\end{theorem}

\section{Additional Experiments}
\label{sec:gnp_exp}

\subsection{Random Networks}

We evaluate our results in simulated $\gnp$ and power law
graphs, analyzing the competitive ratio; we compare a baseline (no intervention), 
expected degree interventions, the first eigenvector intervention (computed from the realized graph), and the optimal intervention for each graph (computed via quadratic programming).

For all experiments, we fix $n=100$ and $b_i = 1$ for each agent.
For both graph families, we experiment by independently varying the budget size $C$, a distribution parameter ($p$ or $\sigma$), and the spectral radius of $\beta A$. We plot the competitive ratio of each heuristic intervention with the optimal intervention as parameters are varied, and each parameter specification is averaged over 10 graph samples.
We generate power law graphs according to the $\gw$ model with a power law sequence, where the maximum expected degree is fixed at 25 and the minimum is fixed at 1 for all exponent values.

\begin{figure}[ht]
    \begin{center}
     \includegraphics[width=5cm]{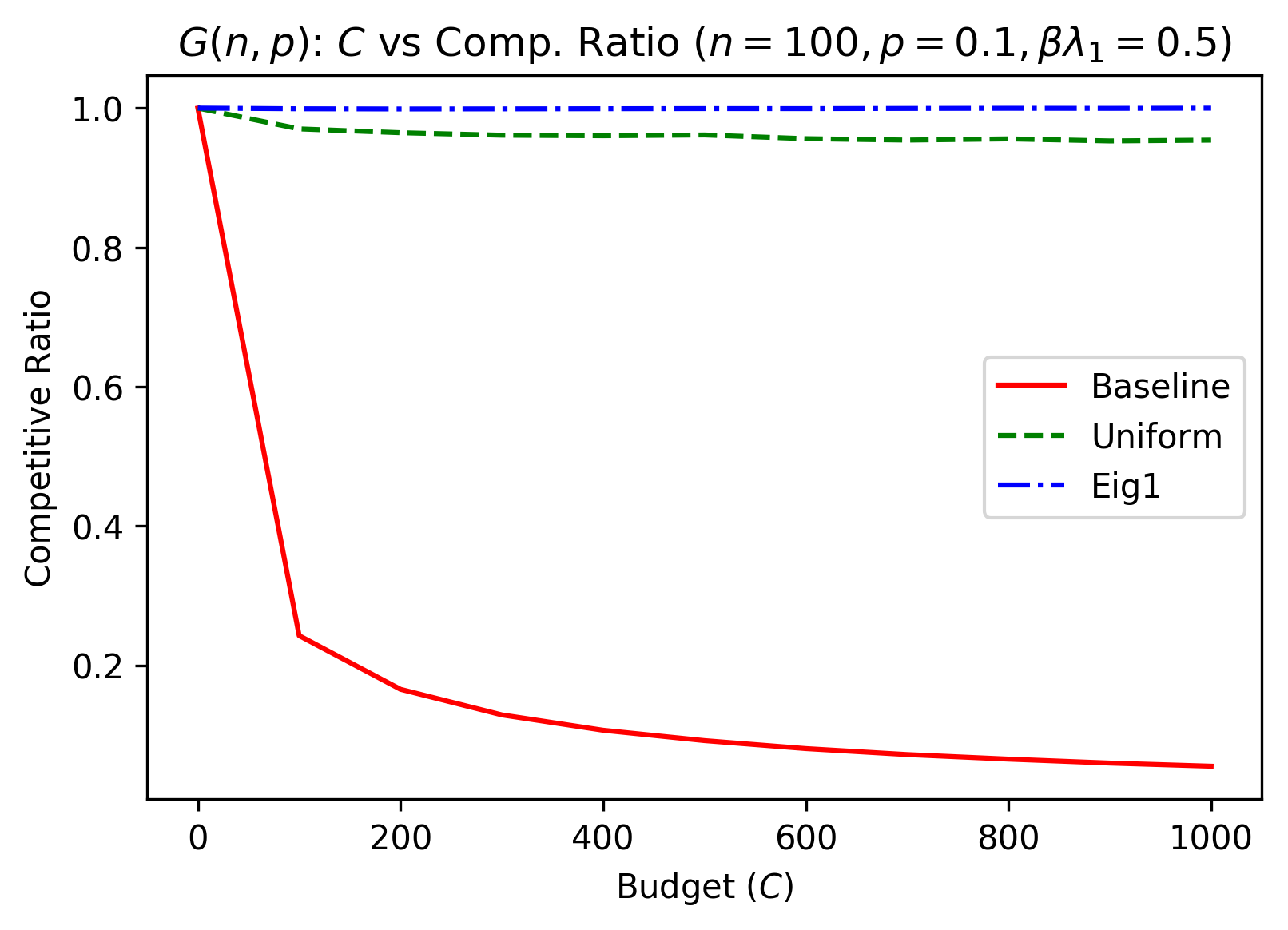}
    \includegraphics[width=5cm]{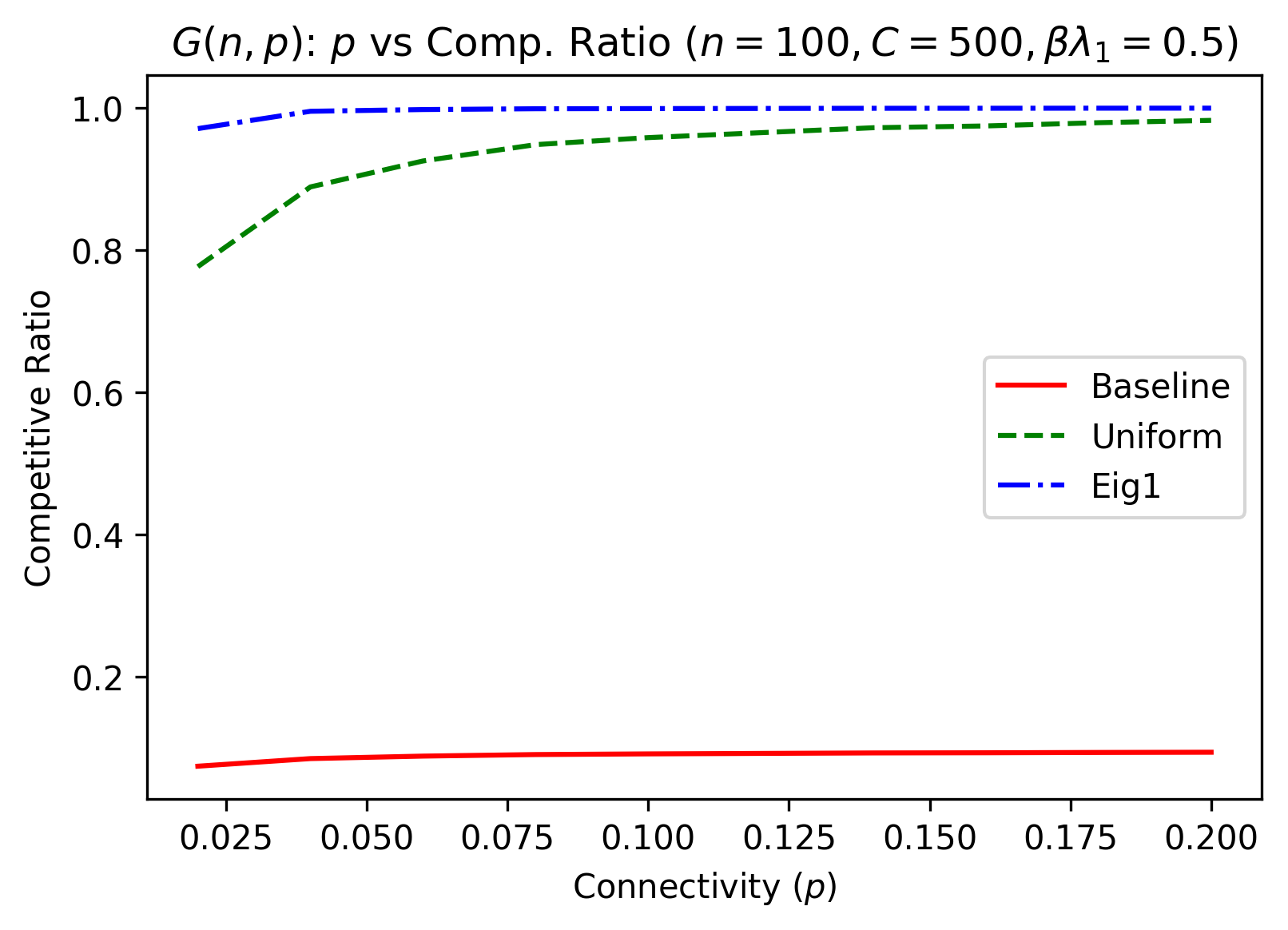}
    \includegraphics[width=5cm]{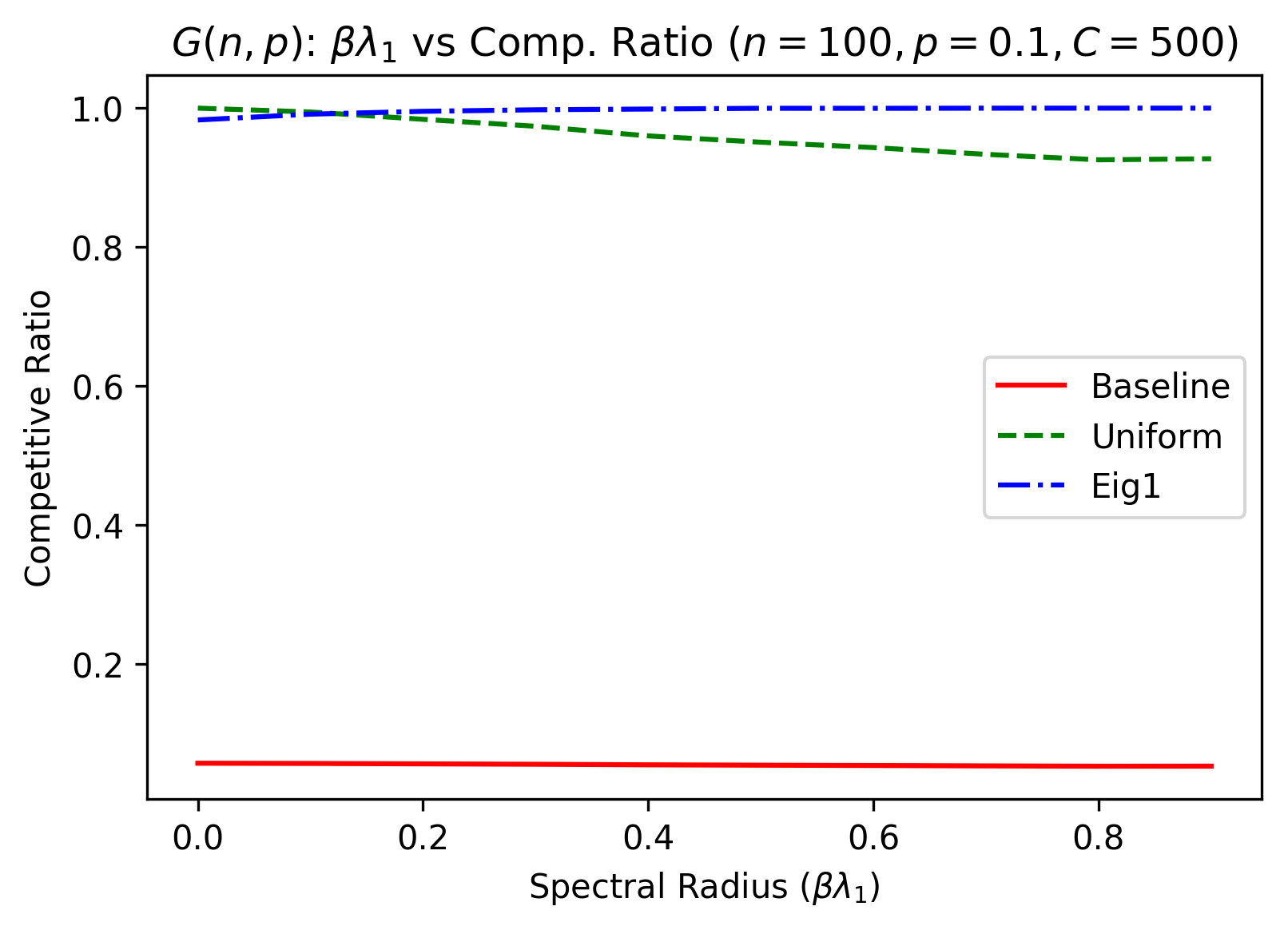}      
    \end{center}
    \caption{Intervention in $\gnp$ Graphs}
    \label{gnp-plots}
    \centering
\end{figure}

In $\gnp$ graphs (see \Cref{gnp-plots}), we see that the first eigenvector intervention is close to optimal in almost all cases, and that the uniform intervention is quite competitive as well, particularly when the graph is increasingly dense. When the spectral radius is small, the uniform outperforms the first eigenvector intervention as expected, it is optimal when $\beta$ is 0. The small baseline values at most points indicate that the change to social welfare from our interventions is indeed quite drastic, even when $C = \norm{b}^2$.

\begin{figure}[ht]
    \begin{center}
    \includegraphics[width=5cm]{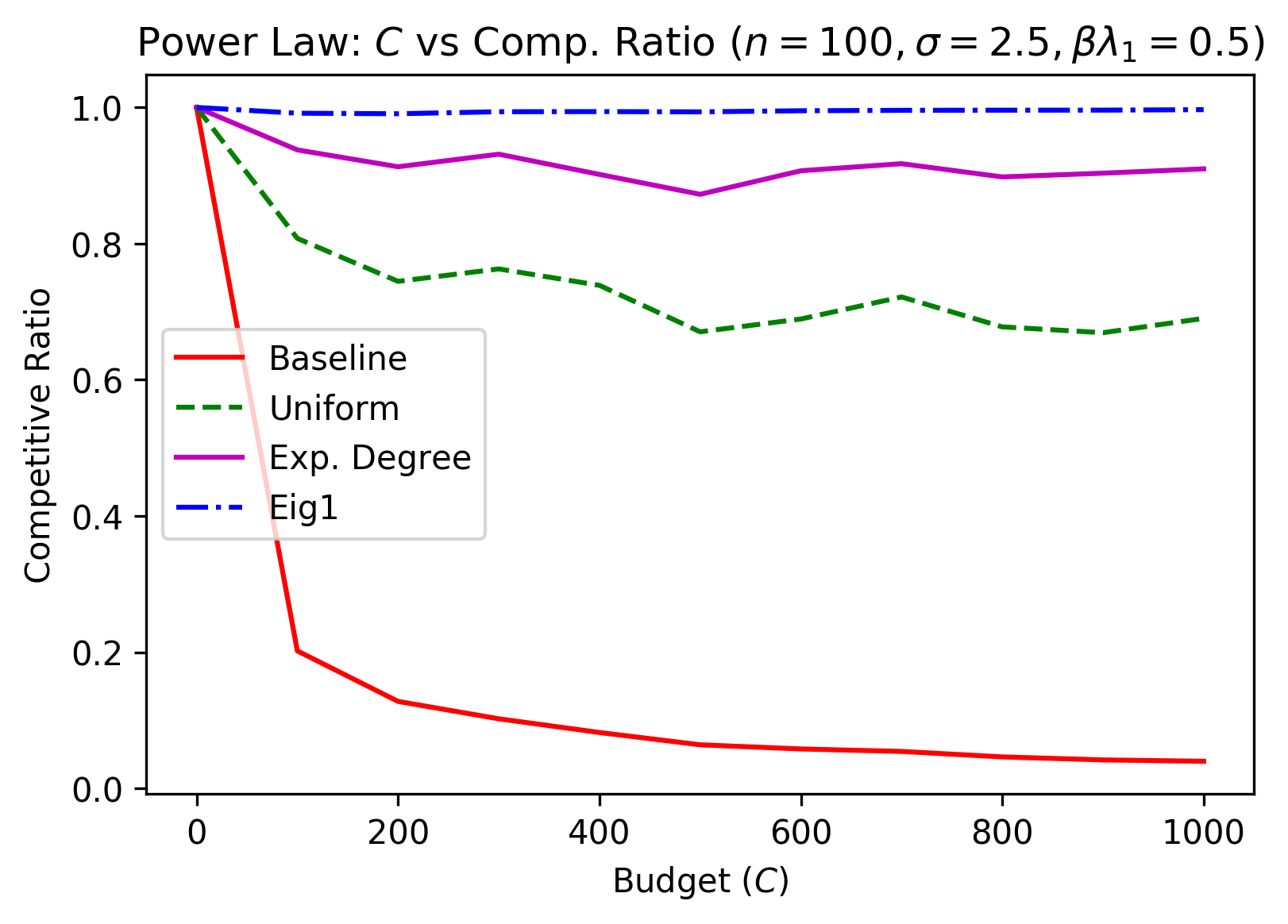}
    \includegraphics[width=5cm]{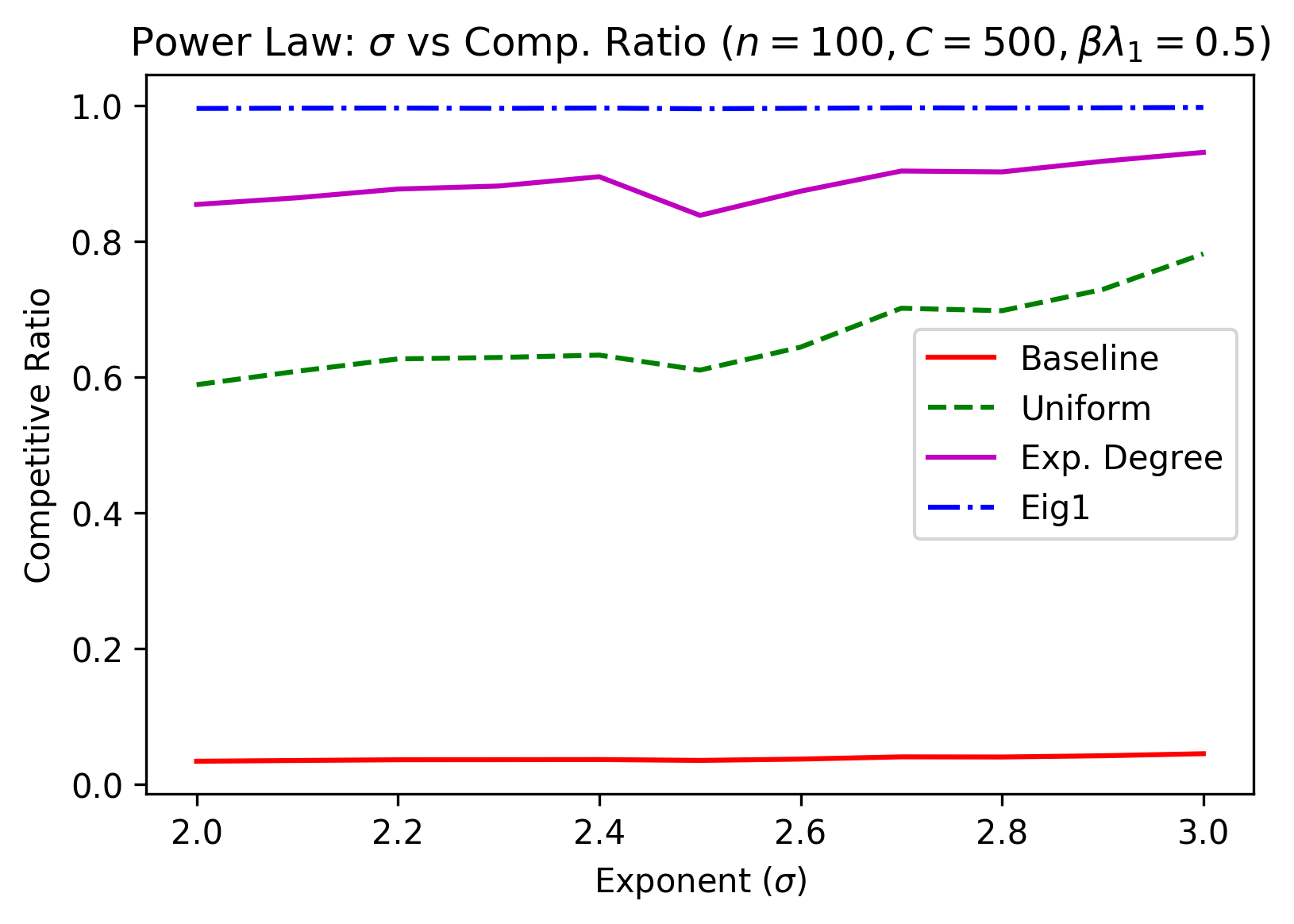}
    \includegraphics[width=5cm]{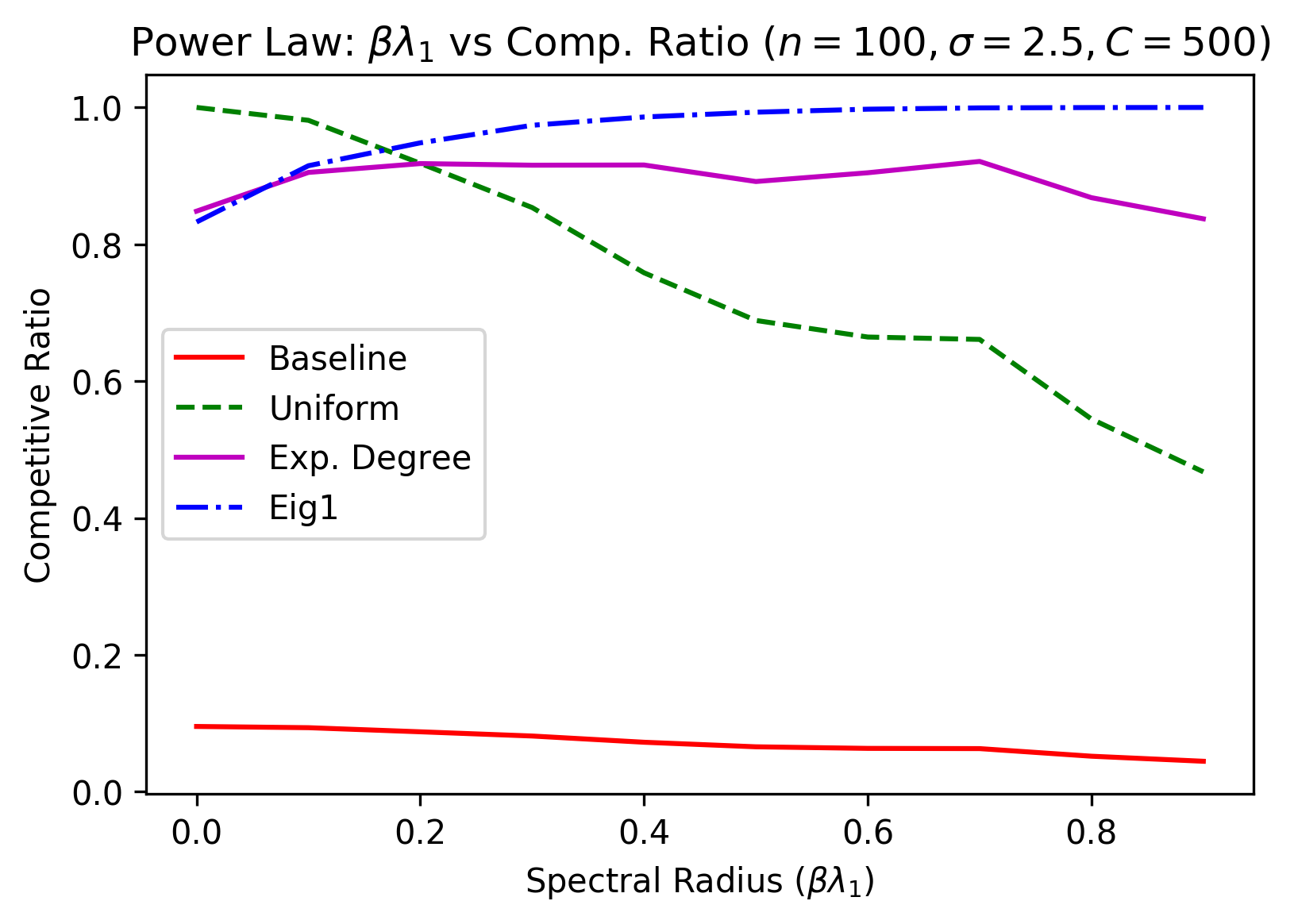}
    \end{center}
    \caption{Intervention in Power Law Graphs}
    \label{gw-plots}
    \centering
\end{figure}

In power law graphs (see \Cref{gw-plots}), the uniform intervention does not perform nearly as well unless the spectral radius is small, where it outperforms other approaches. The expected degree intervention does considerably better in general. When $\sigma > 2.5$, where our theoretical results do not hold, we see that heuristics still perform well. 
We expect that this is an artifact of the fixed expected degree bounds and the small graph size. In large power law graphs, larger exponents correspond to the graph having a smaller ``core'' of dense connectivity, which can be quite important for influencing the rest of the network. Our experiments suggest that smaller cores are less important in graphs of this size, which more closely resemble sparse $\gnp$ graphs with a few well-connected vertices.  

\section{Omitted Proofs}
\label{omittedproofs}

First, we prove a result about general matrices that we will use later
in the proofs.

\begin{lemma}
  For any two $n\times n$ symmetric matrices $A$ and $B$, let $\lambda_1$ and $\lambda_2$ be the largest (according to magnitude) eigenvalues of $A$, and let $\mu_1$ be the largest (according to magnitude) eigenvalue of $B$ (Note: By definition, $||B||=\mu_1$ and $||A||=\lambda_1$).
  Let $\kappa$ be the inverted spectral gap, i.e.\  $|\lambda_2| \leq \kappa|\lambda_1|$. Then,
 
  \[
      \norm{v(A)-v(B)}^2 \leq 2 \left( 1-\frac{1}{\lambda_1}\sqrt{\frac{\left( {\mu_1-\eta} \right)^2-\lambda_1^2\kappa^2}{1-\kappa^2}}  \right)
  \]
  Where $v(\cdot)$ is the unit eigenvector with the largest absolute eigenvalue of the matrix $\cdot$, and $\eta=\norm{A-B}$.
  \label{thm:concentration}
\end{lemma}
\begin{proof}
  First, we observe that,
  \begin{equation}
   \norm{Av(B)} \ge \norm{Bv(B)} - \norm{B-A} \cdot \norm{v(B)} = \mu_1-\eta, \label{eqn:image-lower-bdd}
  \end{equation}
  by the triangle inequality. Also, we can write, $v(B) = \phi_1v(A) + \phi_2w$ where $w \perp v(A)$ is a unit vector, and $\phi_1^2+\phi_2^2 = 1$. Since $w\perp v(A)$, its image under $A$ can at most have a magnitude of $\lambda_2$. That is, 
  \begin{equation*}
    ||Aw||\le \lambda_2
  \end{equation*}
  From our assumption, we can write 
  \begin{align*}
      \norm{Av(B)} &\leq \sqrt{\lambda_1^2\phi_1^2 + \lambda_2^2\phi_2^2} \\
    &\leq \sqrt{\lambda_1^2\phi_1^2 + \kappa^2\lambda_1^2\phi_2^2}\\
    &= \lambda_1\sqrt{\phi_1^2 + \kappa^2\phi_2^2}
  \end{align*}
  Combining with \eqref{eqn:image-lower-bdd}, we get 
  \begin{align*}
    \frac{\mu_1-\eta}{\lambda_1} &\le \sqrt{\phi_1^2 + \kappa^2\phi_2^2} \\
    \implies \left( \frac{\mu_1-\eta}{\lambda_1} \right)^2 &\le \phi_1^2 + \kappa^2(1-\phi_1^2) \\
    \implies \phi_1 &\geq \frac{1}{\lambda_1}\sqrt{\frac{\left( {\mu_1-\eta} \right)^2-\lambda_1^2\kappa^2}{1-\kappa^2}}
  \end{align*}
  Using the fact that $\norm{v(A)-v(B)} = \norm{(1-\phi_1)v(A)+\phi_2w} = \sqrt{(1-\phi_1)^2 + \phi_2^2}=\sqrt{2(1-\phi_1)} $, we get 
  \begin{equation*}
      \norm{v(A)-v(B)}^2 \leq  2 \left( 1-\frac{1}{\lambda_1}\sqrt{\frac{\left( {\mu_1-\eta} \right)^2-\lambda_1^2\kappa^2}{1-\kappa^2}}  \right).
  \end{equation*}
\end{proof}

\begin{proof}
    [Proof of \Cref{thm:sbm-main}]
From \Cref{thm:concentration}, for matrices $\ba$ and $A$, we have that
\begin{align*}
    \norm{v_1(\ba) - v_1(A)}^2 \leq&\; 2 \left( 1 - \frac{1}{\lambda_1}\sqrt{\frac{(\mu_1 - \eta)^2 - \lambda_1^2 \kappa^2}{1-\kappa^2}} \right) 
\end{align*}
where $\lambda_1$ is the first eigenvalue of $\ba$, $\abs{\lambda_i(\ba)} \leq \kappa \lambda_1$ for all $i > 1$, $\mu_1$ is the first eigenvalue of $A$, and $\eta = \norm{\ba - A}$. By the triangle inequality $\mu_1 - \eta \geq \lambda_1 - 2\eta$, and so
\begin{align*}
  \norm{v_1(\ba) - v_1(A)}^2 \leq&\; 2 \left( 1 - \frac{1}{\lambda_1}\sqrt{\frac{(\lambda_1 - 2\eta)^2 - \lambda_1^2 \kappa^2}{1-\kappa^2}} \right) \\  
  \leq&\; 2 \left( 1 - \frac{1}{\lambda_1}\sqrt{\frac{\lambda_1^2(1 - \kappa^2) - 4\eta\lambda_1}{1-\kappa^2}} \right) \\  
  \leq&\; 2 \left( 1 - \sqrt{1 - \frac{ 4\eta }{(1-\kappa^2) \lambda_1}} \right) \\  
  \leq&\; \frac{8 \eta}{(1 - \kappa^2)\lambda_1} & \left(\sqrt{1-x} > 1 - x, x \in [0,1]\right).
\end{align*}
Since $d_{\max}\geq \frac 49\log(2n/\delta)$, we have, from the proof of Theorem 1
in \cite{chungradcliffe} that
 $\eta \leq 2 \sqrt{ d_{\max} \log(2n/\delta)}$. If $\lambda_1 \geq \frac{1024 \sqrt{d_{\max} \log(2n/\delta)}}{(1 -\kappa^2)\epsilon^2}$, we have that
\begin{align*}
    \norm{v_1(\ba) - v_1(A)}^2 \leq&\; \frac{16  \sqrt{ d_{\max} \log(2n/\delta)}}{(1 - \kappa^2)\lambda_1} \\
    \leq&\; \epsilon^2 / 64, 
\end{align*}
and so $\norm{v_1(\ba) - v_1(A)} \leq \epsilon/8$. For a sufficiently large budget, by Proposition 2 in \cite{Galeotti2017TargetingII}, 
$\norm{\frac{y^*}{\norm{y^*}} - v_1(A)} \leq \epsilon/8$, where $y^*$ is the optimal intervention.  
By the triangle inequality this gives us that $\norm{\frac{y^*}{\norm{y^*}} - v_1(\ba)} \leq \epsilon/4$.
Applying \Cref{cosineutil} and law of cosines gives a competitive ratio of $1 - \epsilon$.
\end{proof}

\begin{proof}[Proof of \Cref{cosineutil}]
  First we give a general expression of the change in welfare observed at equilibrium after some intervention $y$.   
\begin{align*}
    \Delta W\left(y ; b, A\right) =&\; W\left(y; b, A\right) - W\left(\mathbf{0}; b, A\right) \\
     =&\; \frac{1}{2} \left( \left(M^{-1} y\right)^{\top}\left(M^{-1} y\right) + 2 \left(M^{-1} b\right)^{\top}\left(M^{-1} y\right)\right) \\
     =&\; \frac{1}{2} \left( \norm{M^{-1} y}^2 + 2 \left(M^{-1} b\right)^{\top}\left(M^{-1} y\right)\right) 
\end{align*}
By the law of cosines, there is some vector $y'$ with norm $\sqrt{2C
  \left(1 - \gamma\right)}$ such that $y^* - y = y'$. 
Let $\lambda_M = \lambda_1(M^{-1})$.
We can then give an additive bound on the utility loss $\Delta U =  W\left(y^*; b, A\right) - W\left(y; b, A\right)$:
\begin{align*}
    \Delta U =&\; \Delta W\left(y^* ; b, A\right) - \Delta W\left(y; b, A\right) \\
    =&\; \frac{1}{2} \left( \norm{M^{-1} y^*}^2 + 2 \left(M^{-1} b\right)^{\top}\left(M^{-1} y^*\right)\right) \\
     &\; \quad - \frac{1}{2} \left( \norm{M^{-1} y}^2 + 2 \left(M^{-1} b\right)^{\top}\left(M^{-1} y\right)\right) \\
     =&\; \frac{1}{2}\left(\norm{M^{-1} y^*}^2 - \norm{M^{-1} \left(y^* - y'\right)}^2\right)  & (\text{substituting } y=y^*-y')\\
     &\; \quad + \left(M^{-1}b\right)^{\top}  M^{-1}\left(y^* - y\right) \\
    \leq&\; \frac{1}{2}\left(\norm{M^{-1} y^*}^2 - \left(\norm{M^{-1} y^*} -  \norm{ M^{-1}y'}\right) ^2\right)\\ 
        &\; \quad + \left(M^{-1}b\right)^{\top}  M^{-1}\left(y^* - y\right) & \left(\text{triangle inequality}\right) \\
    \leq&\; \norm{M^{-1} y^*}\cdot \norm{ M^{-1}y'} + \left(M^{-1}b\right)^{\top}  M^{-1}\left(y^* - y\right) \\
    \leq&\; \lambda_M^2 C \cdot \sqrt{2 \left(1 - \gamma\right) } + \norm{ M^{-1}b} \cdot \norm{ M^{-1}\left(y^* - y\right)}& \text{(Cauchy-Schwarz)} \\
    \leq&\; \lambda_M^2 C \cdot \sqrt{2 \left(1 - \gamma\right) } + \lambda_M \norm{b} \cdot \norm{M^{-1} \left(y^* - y\right)} & \text{(spectral rad.\ def.)}\\
    =&\; \lambda_M^2 C \cdot  \sqrt{2 \left( 1 - \gamma \right) } + \lambda_M^2 \norm{b} \cdot \sqrt{2 C \left(1 - \gamma\right)} & \text{(law of cosines)}
\end{align*}
From our lower bound on $C$
we then have:
\begin{align*}
    W(y^*; b, A) - W(y; b, A)  
    \leq&\; 2 \lambda_M^2 C \sqrt{2(1 - \gamma)} &   \\
\end{align*}
Further, we have that
\begin{align*}
    W(y^*; b, A) \geq&\; \frac{1}{2}\lambda_M^2 C,
\end{align*}
as this is obtainable even if $b=\mathbf{0}$ by letting $y^*$ be the first eigenvector of $M^{-1}$ associated with the spectral radius $\lambda_M$. Recall that social welfare at equilibrium is increasing in $b$. 
We can then give a multiplicative bound on utility loss:
\begin{align*}
    \frac{W(y^*; b, A) - W(y; b, A)}{W(y^*; b, A)}  \leq&\; \frac{2 \lambda_M^2C\sqrt{2(1 - \gamma)} }{\frac{1}{2}\lambda_M^2C} \\
    =&\;4\sqrt{2(1 - \gamma)}
\end{align*}
and so
\begin{align*}
    \frac{ W(y; b, A)}{ W(y^*; b, A)} \geq&\; 1 - 4\sqrt{2( 1-\gamma)}.  
\end{align*}
\end{proof}

\begin{proof}[Proof of \Cref{lemma:eigen-proportionality}]
Let $v$ be an eigenvector of $\ba$, and let $\lambda$ be the corresponding eigenvalue. 
It suffices to show that 
for any $v$ and any $i, k$, 
if $\lambda \neq 0$, then:
\begin{equation}
  \frac{v_i}{v_k} = \frac{w_i}{w_k}.
\label{eqn:power-law-eigenvector-proportionality}
\end{equation}
In $\gw$ graphs, an edge between
vertices $i$ and $j$ is constructed with a probability
proportional to $w_iw_j$. As such, for any $i,j,k$ it holds that:
\begin{equation}
  \frac{\ba_{ij}}{\ba_{kj}} = \frac{w_i}{w_k}. 
  \label{eqn:expected-power-law}
\end{equation}
By the eigenvalue equation, for any $i$,
\begin{align*}
    \sum_{j=1}^n{\ba_{ij}v_j} =&\; \lambda v_i,
\end{align*}
and so for any $i, k$, by \eqref{eqn:expected-power-law},
\begin{align*}
    \lambda \cdot \frac{w_k}{w_i} v_i =&\; \frac{w_k}{w_i}\sum_{j=1}^{n}\overline{A}_{ij}v_j\\
    =&\; \sum_{j=1}^{n}\overline{A}_{kj}v_j \\
    =&\;\lambda v_k 
  \end{align*}
which implies 
\eqref{eqn:power-law-eigenvector-proportionality} when $\lambda$ is non-zero.
\end{proof}

\begin{proof}[Proof of \Cref{lemma:simple-eigenvector-bound}]
 This follows from \Cref{thm:concentration} with $\kappa=0$. We also give a simpler direct alternate proof here, which may be of independent interest.
  
 First, we observe that
  \begin{equation}
   \norm{A v(B)} \ge \norm{Bv(B)} - \norm{B-A} \cdot \norm{v(B)} = \mu_1-\eta, \label{eqn:image-lower-bdd2}
  \end{equation}
  by the triangle inequality. We can write $v(B) = \phi_1 v(A) + \phi_2 x$, where $x$ is a unit vector orthogonal to $v(A)$, and $\phi_1^2+\phi_2^2 = 1$, for some values $\phi_1$ and $\phi_2$. Since $x \perp v(A)$, its image under $A$ can at most have a magnitude of $\lambda_2 = 0$. 
That is, 
\begin{equation*}
    \norm{Ax} = 0. 
\end{equation*}
From our assumption, we can write 
\begin{align*}
    \norm{Av(B)} =&\; \norm{\phi_1 A v(A) + \phi_2 A x } \\
    =&\; \norm{\phi_1 A v(A)} \\
    =&\; \lambda_1 \phi_1. 
\end{align*}
Combining with \eqref{eqn:image-lower-bdd2}, we get that 
\begin{align*}\frac{\mu_1 - \eta}{\lambda_1} \leq&\; \phi_1.
\end{align*}
We can then conclude that
\begin{align*}
    \norm{v(A)-v(B)} =&\; \norm{(1-\phi_1)v(A)+\phi_2w} \\
    =&\; \sqrt{(1-\phi_1)^2 + \phi_2^2} \\
    =&\; \sqrt{1 - 2\phi_1 + (\phi_1^2 + \phi_2^2)} \\ 
    =&\; \sqrt{2(1-\phi_1)} \\
    \leq&\; \sqrt{2\left(1 - \frac{\mu_1 - \eta}{\lambda_1} \right)}.
\end{align*}
\end{proof}

\begin{proof}[Proof of \Cref{lemma:1st-eigenvector-close}]
Let $D$ be the matrix that has the same diagonal entries as $\ba$ and
$0$ everywhere else. 
    First we observe that      
\begin{align*}
  \norm{\ba  - (\ba - D)} \leq&\; 1, 
\end{align*}
as all entries of $D$ are bounded by 1.
Whether or not we are deleting self-loops, by applying the triangle inequality to the previous observation and \Cref{lemma:matrix-norm-deviation},
\begin{align*}
    \norm{A - \ba} \leq&\; \sqrt{4 w_1 \log(2n/\delta)} + 1,
\end{align*}
which implies again by triangle inequality that
\begin{align*}
    \norm{A} \geq&\; \tilde{d} - \sqrt{4 w_1 \log(2n/\delta)} - 1. 
\end{align*}
By our distributional assumption,
\begin{align*}
    \tilde{d} \geq&\; \frac{256 \left( \sqrt{4w_1 \log(2n/\delta)} + 1\right)  }{\epsilon^2}
\end{align*}
By \Cref{lemma:simple-eigenvector-bound}, it follows that
\begin{align*}
    \norm{v_1(A) - v_1(\ba)} \leq&\; \sqrt{2\left(1 - \frac{\norm{A} - \norm{A - \ba} }{\norm{\ba}} \right)} \\
    =&\;   \sqrt{2\cdot \frac{\norm{\ba} - \norm{A} + \norm{A - \ba} }{\norm{\ba}}}    \\
    \leq&\; 2\sqrt{\frac{\norm{A - \ba}  }{\norm{\ba}}} \\ 
    \leq&\; 2\sqrt{\frac{\sqrt{4 w_1 \log(2n/\delta)} + 1}{\tilde{d}}} \\
    \leq&\; 2\sqrt{\frac{\epsilon^2}{256}}\\
     =&\; \epsilon/8.
\end{align*}
\end{proof}

\begin{proof}[Proof of \Cref{lemma:2nd-eigenvalue-bound}]
  The second-largest absolute eigenvalue can be defined as 
  \begin{align}
    \lambda_2 &= \sup_{\substack{x\perp v_1\left(\ba-D\right)\\ \norm{x}=1}} \norm{\left(\ba-D\right)x} \notag \\
    &\leq \sup_{\substack{x\perp v_1\left(\ba-D\right)\\ \norm{x}=1}}\norm{\ba x} + \sup_{\substack{x\perp v_1\left(\ba-D\right)\\ \norm{x}=1}}\norm{Dx}
    \label{eqn:u2-intermidiate}
  \end{align}
  Any vector $x$ that is perpendicular to $v_1\left( \ba-D \right)$ can be written as 
  \begin{equation}
    x = \phi_1v_1\left( \ba \right) + \phi_2u \notag
  \end{equation}
  Where $u$ is perpendicular to $v_1\left( \ba \right)$, and $\phi_1^2+\phi_2^2=1$. Let $d:=v_1\left( \ba-D \right)-v_1\left( \ba \right)$. Since $\ba$ has only one non-zero eigenvalue, we can apply \Cref{lemma:simple-eigenvector-bound} to get 
  \begin{equation}
    \norm d \leq \sqrt{2\left( 1-\frac{\mu_1-\eta}{\lambda_1} \right)} 
    \label{eqn:bound-on-d}
  \end{equation}
  Where $\mu_1=\norm{\ba-D},\;\lambda_1=\norm{\ba},\;\eta=\norm D$. We have, from the triangle inequality,
\begin{equation}
  \quad\quad\norm{\ba }-\norm{D} \quad\leq\quad \norm{\ba -D}  \notag
\end{equation}
Or, 
\begin{align*}
  \lambda_1 - \eta &\le \mu_1 \\
  \implies 1-\frac{\mu_1-\eta}{\lambda_1} &\leq \frac{2\eta}{\lambda_1}
\end{align*}
Putting the values for $\eta=\frac{w_1}{\sum_{i=1}^{n}w_i  },\;\lambda_1 = \frac{\sum_{i=1}^{n}w_i^2  }{\sum_{i=1}^{n}w_i  }$, and using \eqref{eqn:bound-on-d}, we get 
  $$\norm d \leq \frac{2w_1}{\norm w}$$
  Since, $x$ is, by definition perpendicular to $v_1\left( \ba-D \right)$, we can write
  \begin{align*}
    v_1\left( \ba-D \right)\cdot x = 0 &\implies d\cdot x + v_1\left( \ba \right)\cdot x = 0 \\
                                       &\implies \phi_1 = v_1\left( \ba \right)\cdot x = -d\cdot xb
  \end{align*}
  Using Cauchy-Schwarz, we get that $|\phi_1|\leq \norm{d} \leq \frac{2w_1}{\norm w}$. Thus, 
  \begin{align*}
    \ba x &= \phi_1\lambda_1v_1\left( \ba \right) + \phi_2\cdot 0 & \left(\because u\perp v_1\left(\ba\right)\implies u\in\ker\left( \ba \right)\right)\\
    \implies \norm{\ba x} &\leq \frac{2w_1\lambda_1}{\norm w}
  \end{align*}
  Combining with \eqref{eqn:u2-intermidiate}, and noting the fact the eigenvalues of $D$ are all less than $1$, we get the result.
\end{proof}

\begin{proof}[Proof of \Cref{thm:gw-main}] 
    If the stated conditions hold, then with probability at least $1 - \delta$, all desired bounds on eigenvalues and norms are obtained simultaneously by \Cref{thm:chung-radcliffe-main}. 
First we see that this implies a constant-factor separation between $\lambda_1$ and $\lambda_2$ for the realized adjacency matrix. This ensures that $\beta \lambda_1 < 1$ by our condition on $\beta$.
From the proof of \Cref{lemma:1st-eigenvector-close} we know that
\begin{align*}
    \lambda_1 \geq&\; \tilde{d} - \sqrt{4 \w_1 \log(2n/\delta) } - 1,
\end{align*}
and from \Cref{lemma:2nd-eigenvalue-bound-2}, we know that
\begin{align*}
    \lambda_1 \leq&\; \frac{2 w_1 \tilde{d}}{\norm{w}} + \sqrt{4 \w_1 \log(2n/\delta) } +1 \\
    \leq&\;  \frac{\tilde{d}}{3} + \sqrt{4 \w_1 \log(2n/\delta) } +1.
\end{align*}
For any $\epsilon < 1$ and $\delta < \frac{1}{2}$ this implies that $\lambda_1 > \frac{5}{6} \tilde{d}$ and $\lambda_2 < \frac{5}{12} \tilde{d}$, and so $\frac{\lambda_1}{2} > \lambda_2$.
We can use this to show a bound on the spectral gap of $M^{-1}$ which allows us to apply \Cref{prop:GGG-opt}: 
\begin{align}
    \left( \frac{\alpha_2}{\alpha_1 - \alpha_2} \right)^2 
    =&\; \left( \frac{1/(1 - \beta \lambda_2)^2 }{1/(1 - \beta \lambda_1)^2  - 1/(1 - \beta \lambda_2)^2 } \right)^2 \notag \\
    =&\; \left( \frac{1}{\left(\frac{1- \beta \lambda_2}{1 - \beta \lambda_1}\right)^2 - 1 } \right)^2 \notag \\ 
    \leq&\; \left( \frac{1}{\left(\frac{1- 0.5 \beta \lambda_1}{1 - \beta \lambda_1}\right)^2 - 1 } \right)^2 & (\lambda_1 \geq 2 \lambda_2 ) \notag \\ 
    \leq&\; \frac{1}{(\beta \lambda_1)^2}. \label{eq:budgetbound}
\end{align}
where the final inequality holds by noting that equality holds only at $\beta \lambda_1 = 1.25$, and that \eqref{eq:budgetbound} is larger for all $\beta \lambda \in (0,1)$.
Thus, if 
\begin{align*}
    C \geq&\; \frac{2\norm{b}^2}{(\epsilon/(8\sqrt{2}))^2} \cdot \left( \frac{\alpha_2}{\alpha_1 - \alpha_2} \right)^2 \\
    =&\; \frac{256 \norm{b}^2}{\left( \beta \tilde{d} \epsilon\right) ^2}, 
\end{align*}
then $\rho(y^*, v_1(A)) \geq \sqrt{1 - (\epsilon/(8\sqrt{2}))^2} $ by \Cref{prop:GGG-opt}. 
Applying the inequality $\sqrt{1 - x} \geq 1 - x$ for $x \in [0,1]$, we have that $\rho(y^*, v_1(A)) > 1 - (\epsilon/(8\sqrt{2}))^2 = 1 - \epsilon^2/128$. 
Scaling $y^*$ to the unit vector $\frac{y^*}{\norm{y^*}}$ and applying the law of cosines, we get that $\norm{\frac{y^*}{\norm{y^*}} - \frac{v_1(A)}{\norm{v_1(A)}}} \leq \epsilon/8$.  
From \Cref{lemma:1st-eigenvector-close}, we also have that $\norm{v_1(A) - \frac{w}{\norm{w}}} \leq \epsilon/8$.  
By the triangle inequality this gives us that $\norm{\frac{y^*}{\norm{y^*}} - \frac{\w}{\norm{\w}}} \leq \epsilon/4$.
Applying \Cref{cosineutil} and law of cosines gives a competitive ratio of $1 - \epsilon$.
\end{proof}

\begin{proof}[Proof of \Cref{lemma:gnp-eigenvector-bound}]
    When considering  $\gnp$ graphs where self-loops are added with probability $p$, this is equivalent to the $\gw$ graph distribution with $\w_i = np$ for all $i$. We give sufficient conditions for applying \Cref{lemma:1st-eigenvector-close}, which holds even if we do not allow self-loops. We can compute the second-order average expected degree as: 
\begin{align*}        
   \tilde{d} =&\; \frac{ \sum_i w_i^2  }{\sum_i w_i } \\
    =&\; \frac{n^2p^3}{np^2} \\ 
\end{align*}
and then can see that  
\begin{align*}
    \frac{256\left( \sqrt{4 w_1 \log(2n/\delta) } + 1 \right)}{\epsilon^2} \leq&\; \frac{256 \sqrt{6 w_1 \log(2n/\delta) } }{\epsilon^2} \\ 
    \leq&\; \frac{768 \sqrt{6} }{\epsilon^2} \cdot \sqrt{n} p,  & (p \geq \frac{4}{9} \log(2n / \delta))
\end{align*}
and so whenever $n$ is at least $(786^2 \cdot 6  /\epsilon^4)$,
we will have that $\tilde{d} = np$ is sufficiently large.
We already have the required bound on the maximum expected degree, which in fact holds for all expected degrees, and so we can directly apply \Cref{lemma:1st-eigenvector-close} to obtain the result.   
\end{proof}

\begin{proof}[Proof of \Cref{thm:gnp-main}]
    In order to apply \Cref{thm:gw-main}, we simply need to show that the spectral gap is not too small, which previously followed from the condition that $w \leq \frac{\norm{w}}{6}$. If $n$ is at least 36, this will hold for the $\gw$ interpretation of $\gnp$ graphs (prior to removing the diagonal), as $\norm{w} = n^{1.5} p$, so then $\frac{\norm{w}}{w} = \sqrt{n} \geq 6$. 
\end{proof}

\begin{proof}[Proof of \Cref{lemma:power-law-eig}]
   In \cite{Chung2003} it is proven that whenever $\sigma\in(2,2.5)$, the expected largest eigenvalue of a random power law graph is $\Theta \left( w_1^{3-\sigma} \right)$, where $w_1$ is the highest degree.     
   If $w_1$ is at least 
$\Omega\left( \left( \frac{\log(2n/\delta}{\epsilon^4} \right)^{\frac{1}{5 - 2\sigma}} \right)$, 
then the largest eigenvalue will be $\Omega \left( \frac{ \sqrt{w_1 \log n}}{\epsilon^2} \right)$. 
With appropriate constants, we obtain the required bound on $\tilde{d} = \lambda_1(\ba)$ to ensure that deviations from expectation are sufficiently small such that we can apply \Cref{lemma:1st-eigenvector-close}.  
\end{proof}

\begin{proof}[Proof of \Cref{thm:power-law-main}]
    By \Cref{lemma:power-law-eig}, the first eigenvector of the graph is sufficiently close to its expectation.
    For power law graphs larger than some constant, $\frac{w}{\norm{w}} < \frac{1}{6}$ will hold,
and so the second eigenvalue will be small enough such that, when applying Theorem 1 from \cite{chungradcliffe},
a spectral gap factor of $1/2$ will be obtained almost surely.
The bound above gives us 
a constant factor spectral gap for sufficiently large constants, 
which allows us to obtain the same budget lower bound as in \Cref{thm:gw-main}.  
We can then combine these results as in \Cref{gw} to obtain the theorem.
\end{proof}

\begin{proof}[Proof of \Cref{lemma:util-upper-bound}]
  If cosine similarity is small, much of the mass of the intervention
  is in components orthogonal to $y^*$, whose efficacy can be
  upper-bounded by the second eigenvalue of $M^{-1}$. This bound will
  depend on the first and second eigenvalues of $\beta A$. For some
  $\phi_1$, $\phi_2$ where $\phi_1^2 + \phi_2^2 = 1$, and for
  some unit vector $u$ orthogonal to $y^*$, we can write:
\begin{align*}
    y &= \phi_1 y^* + \phi_2 \sqrt{C} u \\
  \therefore W(y) &= \norm{M^{-1}y}^2 = \norm{\phi_1M^{-1}y^* +
                        \phi_2\sqrt CM^{-1}u}^2\\
     &= \phi_1^2\norm{M^{-1}y^*}^2 + C\phi_2^2\norm{M^{-1}u}^2 +
       2\phi_1\phi_2\brack{M^{-1}y^*}^\top \brack{M^{-1}u}\\
  \therefore \frac{W(y)}{W(y^*)} &\leq \phi_1^2 +
                                           \phi_2^2\frac{\alpha_2}{\alpha_1}
                                           + 2\phi_1\phi_2\sqrt{\frac
                                           {\alpha_2}{\alpha_1}}\\
     &\leq \gamma^2\brack{1-\frac{\alpha_2}{\alpha_1}} +
       \frac{\alpha_2}{\alpha_1} + 2\sqrt{\frac{\alpha_2}{\alpha_1}}\\
\end{align*}
\end{proof}

\begin{proof}[Proof of \Cref{lemma:real-deg-near-opt}]
From the $\gw$ version of the $(\epsilon, \delta)$-concentration condition, we have that
\begin{align*}
    \frac{\sqrt{w_1 \log(n / \delta) }}{\epsilon^2} \leq&\; O\parens{\frac{\norm{w}_2^2}{\norm{w}_1}}
\end{align*}
which implies that 
\begin{align*}
    \sqrt{ n w_1 \log(n / \delta) } \leq&\; O\parens{\epsilon^2 {\norm{w}_2}} \\
    \leq&\; O\parens{\epsilon {\norm{w}_2}}
\end{align*}
as $\norm{w}_1 \leq \sqrt{n} \norm{w}_2$ for any vector of length $n$.
The difference in degree vectors $\norm{w - w^*}$ can be expressed as
\begin{align*}
    \norm{\ba \mathbf{1} - A \mathbf{1}} \leq&\; \norm{\ba - A} \cdot \norm{\mathbf{1}}\\
    \leq&\; O\parens{ \sqrt{ n w_1 \log(n / \delta) } },
\end{align*}
where the last line follows from
\Cref{lemma:matrix-norm-deviation} and the $(\epsilon,\delta)$-concentration conditions.
From the analysis of \Cref{thm:sbm-main,thm:gw-main}, we can see that approximating the first eigenvector of $A$ to within $O(\epsilon)$ $\ell_2$ distance is sufficient for near-optimality, which we have for the true first eigenvector by triangle inequality with $w$.
\end{proof}

\begin{proof}[Proof of \Cref{lemma:gw-mixing-time}]

To see that there will be a giant component containing almost all vertices, note that this follows from analysis of the giant component in $G(n,p)$ graphs with $p = \frac{1}{\epsilon n}$, as all probabilities in $\ba$ are at least this large. A $\gnp$ graph with $p > \frac{1}{n}$ almost surely has a giant component of size $\Theta(n)$.
Consider the sampling of all edges in such a graph except for those incident on vertex $i$. The expected number of edges from $i$ to the $\Theta(n) = n^*$ size component will still be $\Theta(1/\epsilon)$. Each edge is included or excluded independently, and so the probability that no edges are included is at most $\parens{ 1 - \frac{O(1/\epsilon)}{n^*} }^{n^*} \leq \exp\parens{-O(1/\epsilon)}$. As this holds for each vertex, the expected number of vertices outside the giant component is $n\cdot \exp\parens{-O(1/\epsilon)}$. $A$ will be a block diagonal matrix with a block for each component, and the largest eigenvalue will be associated with this giant component. We concern ourselves only with this component in our mixing time analysis.

The mixing time of a graph $A$ can be analyzed using the second eigenvalue of its diffusion matrix. 
The first eigenvalue of $P$ will be 1, and from Proposition 1 in \cite{10.1007/BFb0023849}, it follows that random walks on connected graphs
have mixing time of $O\brack{\log (n/\epsilon)}$ to within additive probability $2\epsilon/n$ for each vertex, and thus $\epsilon$ total variation distance, if $\lambda_2(P)$ is bounded by a constant factor below 1. 

Observe that the second eigenvalue of $\overline{P} = \ba \overline{D}^{-1}$ is 0, as $\overline{P}$ has rank 1.
$A$ and $D$ will be close in spectral norm to their expectations for graphs satisfying $(\epsilon,\delta)$-concentration, and the same will hold for $D^{-1}$ when the minimum degree is large enough. 
Note that $n = \Omega\parens{\frac{\sqrt{w_1 \log(n/\delta)}}{\epsilon^2}}$ by the $(\epsilon,\delta)$-concentration conditions, as the first eigenvalue of $\ba$, which will always be less than $n$, must be at least this large. Given the lower bound on $w_1$, we then have that $n = \Omega\parens{\frac{ \log(n/\delta)}{\epsilon^2}}$, with a constant larger than 100.
We can then apply Hoeffding's inequality to see that the true degree of each vertex will be within $\frac{1}{2\epsilon}$ of its expectation, simultaneously with probability at least $1 - O(\delta)$. As such, for any vector $v$, the standard basis components of $D^{-1}v$ and $\overline{D}^{-1}$ are all within a factor of 2. We have already seen that when a $\gw$ graph $\ba$ satisfies $(\epsilon,\delta)$-concentration, $Av$ will be close to the projection of $v$ onto $v_1(A)$, as $\lambda_i(A) \ll \lambda_1(A)$ for $i > 1$. Together with the previous observation, this implies that the image of any vector will be close under $P$ and $\overline{P}$. Thus, with high probability $\lambda_2(P)$ will be bounded by a constant less than 1, giving us our desired mixing time.

\end{proof}

\begin{proof}[Proof of \Cref{lemma:sbm-block-bound}]
    Our proof will proceed using Theorem 5 from \cite{chungradcliffe} for bounding the deviation of sums of Hermitian random matrices. 
$\hat{A}$ can be viewed as a sum of $O(m^2)$ independent matrix random variables, with one variable for each pair of groups.
For each pair of groups, the matrix $\hat{A}$ with zeros for all other group pair entries is symmetric and Hermitian, and the sum of all of these matrices is $\hat{A}$.

    First, we show that each of these individual block matrices $\hat{A}^{ij}$ has an empirical frequency close to its expectation. By considering random variables for each of the $O(n^2/m^2)$ edges and applying Chernoff bounds, we have that $$\abs{ \hat{p}_{ij} - p_{ij} } \leq O\left( \frac{m}{n} \sqrt{\log(m/\delta)} \right)$$
for all groups with probability $1 - \delta/2$.

To apply the theorem, we need to obtain a bound on $\norm{\ba^{ij} - \hat{A}^{ij}}$, the expected and empirical block matrices for each group.
All $O(n^2/m^2)$ non-zero entries in such a matrix equal will be, and will be bounded by $O (m/n \cdot \sqrt{ \log(m/\delta)})$ as we have just seen. The first eigenvector for such a matrix will have $O(n/m)$ non-zero entries, each $O(\sqrt{m/n})$, and its product with the matrix will again have $O(n/m)$ non-zero entries, each of which will be $O(\sqrt{m/n \cdot \log(m / \delta)})$. Thus the spectral norm, equivalent to the corresponding first absolute eigenvalue, will be $O(\sqrt{\log(m/\delta)})$.

Further, we need to bound the norm of the sum of the variance matrices to apply the theorem, which we denote $V^2 = \norm{\sum_{ij} \text{var}(\hat{A}^{ij})}$. For a matrix $A$,  $\text{var}(A) = \E[(A - \E[A])^2]$. 
Each entry in $(\hat{A}^{ij} - \ba^{ij})^2$ will be $O(k(p_{ij} - \hat{p}_{ij})^2)$
where $k$ is the number of edges in the block. 
The expected value of this is simply $k \cdot \text{var}(\hat{p}_{ij}) = p_{ij}(1-p_{ij})/k$.
Summing over all block variance matrices gives us a matrix where all entries (other than the diagonal if self-loops are ignored) are $O(m^2 / n^2)$; this matrix is $n$-dimensional, and so the spectral norm will be $O(m^2 / n)$. 

By Theorem 5 from \cite{chungradcliffe}, with $K = O\left(\max{\left(\frac{m\sqrt{\log(n/\delta)}}{\sqrt{n}}, \log^2(n/\delta)\right)}\right)$ :
\begin{align*}
    \Pr\left[\norm{\hat{A} - \ba} > K\right] \leq&\; 
    O\left(n \cdot \exp{\left( - \frac{K^2}{m^2 / n + K \cdot \log(m/\delta)}\right) } \right) \\
    =&\; O\left(n \cdot \exp{\left( - \frac{K \log(n/\delta )}{K}\right) }\right) \\
        \leq&\; \delta/2. \quad\quad\quad\quad\quad\quad\quad\quad\quad \quad\quad\quad\quad\quad \text{(approp. constants)}
\end{align*}
This completes the proof.
\end{proof}

\section{Analyzing Best Response Dynamics}
\label{best-response}
We can see that the social welfare of the game acts as a potential function when all agents' actions are below equilibrium levels, 
and so we should expect selfish agents to reach the new equilibrium after applying our intervention.

\begin{theorem}[Convergence of Best Response Dynamics]
    \label{thm:potential}
When $\beta > 0$, $b > 0$ and $a_0 < a^*$,
repeated best responses of agents will converge to equilibrium.
\end{theorem}

\begin{proof}[Proof of \Cref{thm:potential}]
Suppose $\beta$ and all $b_i$ values are positive, and all agents start with initial action levels $a_i = 0$. This is below the optimal action level for all agents, and each agent's utility function is concave upon fixing all other action values. When each agent best responds, they will increase their action level towards the optimum value. Because $\beta$ is positive, these increases can only improve the welfare of all other agents. Social welfare increases upon every best response, and thus best response dynamics will converge. The optimal intervention $\beta > 0$ will always be non-negative for the graphs we consider, as the social welfare is an increasing function of $b$. We can again see that Best Response Dynamics would result in asymptotic convergence to the new Nash equilibrium after updating the $b_i$ values; all updates will be non-negative, and thus the social welfare of the network remains a potential function for the game.
If agents only best-respond by more than $\epsilon$ in a given round, convergence will be polynomial in $\frac{1}{\epsilon}$ and other relevant parameters.
\end{proof}

\section{Further Preliminaries}

\subsection{Eigenvalue Transformations}
We note some observations about the eigenvalues of graphs which satisfy the assumptions from Section \ref{prelims}. For a graph $A$, eigenvalue $\lambda_i$ corresponds to eigenvalue $\frac{1}{1- \lambda_i}$ in $(I - A)^{-1}$ 
and $\frac{1}{1 - \beta\lambda_i}$ in $(I - \beta A)^{-1}$; corresponding eigenvectors will be identical for both $A$ and $(I - \beta A)^{-1}$. We will refer to this latter matrix $(I - \beta A)^{-1}$ as $M^{-1}$. Given that the spectral radius of $\beta A$ is less than $1$, $\frac{1}{1-\beta \lambda_i}$ is decreasing in $\lambda_i$, and so the eigenvalues of $M^{-1}$ are ordered according to their corresponding eigenvalues in $A$. Further, all eigenvalues of $M^{-1}$ will be positive. It follows that the spectral radius of $M^{-1}$, which we denote by $\lambda_M$, is $\frac{1}{1 - \beta \lambda_1(A)}$.

\subsection{Imported Results}
\label{sec:imports}
For completeness, we restate important results which we use in our analysis.
A key proposition from \cite{Galeotti2017TargetingII} shows that in the setting we consider, for large enough budgets, the optimal intervention for a graph is close in cosine similarity to the first eigenvector. 

\begin{prop}[Proposition 2 in \cite{Galeotti2017TargetingII}]\label{prop:GGG-opt}
    Suppose $A$ is symmetric, the spectral radius $\lambda_1(A)$ is less than $1/\beta$, and $\beta > 0$. Then for any $\epsilon > 0$, if $$C > \frac{2 \norm{b}^2}{\epsilon} \bigg(\frac{\alpha_2}{\alpha_1 - \alpha_2}\bigg)^2$$ then 
  $\rho(y^*, \sqrt{C}v_1(A)) > \sqrt{1 - \epsilon}$.
\end{prop}

We also state a key theorem about the spectra of random graphs.
\begin{theorem}[Theorem 1 in \cite{chungradcliffe}] \label{thm:chung-radcliffe-main}
      For a random graph with edges constructed independently according to $\bar{A}$ and maximum expected degree $\Delta \geq \frac{4}{9} \ln(2n/\delta)$, with probability $1-\delta$ for sufficiently large $n$, $$|\lambda_i (A) - \lambda_i (\bar{A})| \leq \sqrt{4\Delta \ln(2n/\delta)}$$
for all $1 \leq i \leq n$.
\end{theorem}
As noted in \cite{RePEc:arx:papers:1709.10402}, implicit from the proof of \Cref{thm:chung-radcliffe-main} in \cite{chungradcliffe} (where it is Theorem 1) is a deviation bound on the matrix norm. The result holds for all random graphs with independent edges, and we restate it here for the special case of $\gw$ graphs.

\begin{lemma}[Restatement of Theorem 1 in \cite{chungradcliffe}]    \label{lemma:matrix-norm-deviation}
    For a random graph with independent edges $A$ drawn from $\ba$, if the maximum expected degree $d_{\max}$ is at least $\frac{4}{9}\log(2n/\delta)$, then with probability at least $1 - \delta$ it holds that 
\begin{align*}
    \norm{A - \ba} \leq&\; \sqrt{4 d_{\max} \log(2n/\delta)}.
\end{align*}
\end{lemma}

\end{document}